\documentclass[12pt]{article}

\usepackage{graphicx,psfrag,epsf}
\usepackage{enumerate}
\usepackage{natbib}
\usepackage{mathtools}
\usepackage{booktabs}
\usepackage[english]{babel} 
\usepackage[protrusion=true,expansion=true]{microtype} 
\usepackage{amsmath,amsfonts,amsthm}
\usepackage{chngcntr}
\usepackage[toc,page]{appendix}
\usepackage{ amssymb }
\usepackage{array}
\usepackage{booktabs} 
\usepackage{natbib}
\usepackage{setspace}
\usepackage{soul}
\usepackage{multirow}
\usepackage{enumitem}
\usepackage{microtype} 
\usepackage[font = small,labelfont=bf,textfont=it]{subcaption}
\usepackage{color}
\usepackage{tabularx}
\usepackage[font = small,labelfont=bf,textfont=it]{caption} 
\usepackage{footnote}
\usepackage{algorithm}
\usepackage[noend]{algpseudocode}
\usepackage{etoolbox}

\algblock{ParFor}{EndParFor}
\algnewcommand\algorithmicparfor{\textbf{for}}
\algnewcommand\algorithmicpardo{\textbf{do\ parallel}}
\algnewcommand\algorithmicendparfor{\textbf{end\ parallel\ for}}
\algrenewtext{ParFor}[1]{\algorithmicparfor\ #1\ \algorithmicpardo}
\algrenewtext{EndParFor}{\algorithmicendparfor}

\makeatletter
\DeclareMathOperator*{\argmax}{\arg\!\max}
\newcommand{\distas}[1]{\mathbin{\overset{#1}{\kern\z@\sim}}}%
\newcommand{\bm}[1]{\mathbf{#1}}
\newsavebox{\mybox}\newsavebox{\mysim}
\newcommand{\distras}[1]{%
  \savebox{\mybox}{\hbox{\kern3pt$\scriptstyle#1$\kern3pt}}%
  \savebox{\mysim}{\hbox{$\sim$}}%
  \mathbin{\overset{#1}{\kern\z@\resizebox{\wd\mybox}{\ht\mysim}{$\sim$}}}%
}
\newtheorem{theorem}{Theorem}

\newcommand{\be}{\begin{equation}}
\newcommand{\ee}{\end{equation}}
\newcommand{\bi}{\begin{itemize}}
\newcommand{\ei}{\end{itemize}}
\newcommand{\ben}{\begin{enumerate}}
\newcommand{\een}{\end{enumerate}}
\newcommand{\stb}{\State $\bullet$ \;}

\newcolumntype{K}[1]{>{\centering\arraybackslash}p{#1}}
\allowdisplaybreaks
\DeclareMathOperator*{\argmin}{\arg\!\min}
\makeatother

\let\oldbibliography\thebibliography
\renewcommand{\thebibliography}[1]{\oldbibliography{#1}
\setlength{\itemsep}{0pt}} 

\newcommand{\blind}{1}

\patchcmd{\footnotemark}{\stepcounter{footnote}}{\refstepcounter{footnote}}{}{}
\newcounter{savecntr}
\newcounter{restorecntr}

\begin{document}

\def\spacingset#1{\renewcommand{\baselinestretch}%
{#1}\small\normalsize} \spacingset{1}

\if1\blind
{
  \title{\bf An efficient surrogate model for emulation and physics extraction of large eddy simulations}
  \small
  \author{Simon Mak $^{\text{\textsection}}$ \setcounter{savecntr}{\value{footnote}}\thanks{Joint first authors}\;, Chih-Li Sung\setcounter{restorecntr}{\value{footnote}} $^{\text{\textsection}}$%
  \setcounter{footnote}{\value{savecntr}}\footnotemark
  \setcounter{footnote}{\value{restorecntr}}
   \hspace{.2cm}\\
    and \\		
    Xingjian Wang\setcounter{savecntr}{\value{footnote}}\thanks{School of Aerospace Engineering, Georgia Institute of Technology}, Shiang-Ting Yeh\setcounter{restorecntr}{\value{footnote}}%
  \setcounter{footnote}{\value{savecntr}}\footnotemark
  \setcounter{footnote}{\value{restorecntr}}, Yu-Hung Chang\setcounter{restorecntr}{\value{footnote}}%
  \setcounter{footnote}{\value{savecntr}}\footnotemark
  \setcounter{footnote}{\value{restorecntr}}\\
  and \\
  V. Roshan Joseph$^{\text{\textsection}}$, Vigor Yang$^{\dagger}$, C. F. Jeff Wu\footnote{Corresponding author} \footnote{School of Industrial and Systems Engineering, Georgia Institute of Technology}
}
  \maketitle
} \fi

\if0\blind
{
  \bigskip
  \bigskip
  \bigskip
  \begin{center}
    {\LARGE\bf An efficient surrogate model of large eddy simulations for design evaluation and physics extraction}
\end{center}
  \medskip
} \fi

\bigskip

\vspace{-0.5cm}
\begin{abstract}

In the quest for advanced propulsion and power-generation systems, high-fidelity simulations are too computationally expensive to survey the desired design space, and a new design methodology is needed that combines engineering physics, computer simulations and statistical modeling. In this paper, we propose a new surrogate model that provides efficient prediction and uncertainty quantification of turbulent flows in swirl injectors with varying geometries, devices commonly used in many engineering applications. The novelty of the proposed method lies in the incorporation of known physical properties of the fluid flow as {simplifying assumptions} for the statistical model. In view of the massive simulation data at hand, which is on the order of hundreds of gigabytes, these assumptions allow for accurate flow predictions in around an hour of computation time. To contrast, existing flow emulators which forgo such simplications may require more computation time for training and prediction than is needed for conducting the simulation itself. Moreover, by accounting for coupling mechanisms between flow variables, the proposed model can jointly reduce prediction uncertainty and extract useful flow physics, which can then be used to guide further investigations.
\end{abstract}

\noindent%
\vfill

\newpage
\spacingset{1.45} 

\section{Introduction}\label{sec:intro}

In the quest for designing advanced propulsion and power-generation systems, there is an increasing need for an effective methodology that combines engineering physics, computer simulations and statistical modeling. A key point of interest in this design process is the treatment of turbulence flows, a subject that has far-reaching scientific and technological importance \citep{McC1990}. Turbulence refers to the irregular and chaotic behavior resulting from motion of a fluid flow \citep{Pop2001}, and is characterized by the formation of eddies and vortices which transfer flow kinetic energy due to rotational dynamics. Such a phenomenon is an unavoidable aspect of everyday life, present in the earth's atmosphere and ocean waves, and also in chemically reacting flows in propulsion and power-generation devices. In this paper, we develop a surrogate model, or emulator, for predicting turbulent flows in a swirl injector, a mechanical component with a wide variety of engineering applications.

There are two reasons why a statistical model is required for this important task. First, the time and resources required to develop an effective engineering device with desired functions may be formidable, even at a \textit{single} design setting. Second, even with the availability of high-fidelity simulation tools, the computational resources needed can be quite costly, and only a handful of design settings can be treated in practical times. {For example, the flow simulation of a single injector design takes over 6 days of computation time, parallelized using 200 CPU cores.} For practical problems with large design ranges and/or many design inputs, the use of only high-fidelity simulations is insufficient for surveying the full design space. In this setting, emulation provides a powerful tool for efficiently predicting flows at any design geometry, using a small number of flow simulations as training data. A central theme of this paper is that, by properly \textit{eliciting} and \textit{applying} physical properties of the fluid flow, simplifying assumptions can be made on the emulator which greatly reduce computation and improve prediction accuracy. In view of the massive simulation datasets, which can exceed many gigabytes or even terabytes in storage, such efficiency is paramount for the usefulness of emulation in practice. 

The proposed emulator utilizes a popular technique called \textit{kriging} \citep{Mat1963}, which employs a Gaussian Process (GP) for modeling computer simulation output over a desired input domain. The main appeal of kriging lies in the fact that both the emulation predictor and its associated uncertainty can be evaluated in closed-form. For our application, a kriging model is required which can predict {flows} at any injector geometry setting; we refer to this as \textit{flow kriging} for the rest of the paper. In recent years, there have been important developments in flow kriging, including the works of \cite{Wea2006} and \cite{Rou2008} on {regular spatial grids} (i.e., outputs are observed at the same spatial locations over all simulations), and \cite{Hea2015} on irregular grids. Unfortunately, it is difficult to apply these models to the more general setting in which the \textit{dimensions} of spatial grids vary greatly for different input variables. In the present work, for instance, the desired design range for injector length varies from 20 mm to 100 mm. Combined with the high spatial and temporal resolutions required in simulation, the resulting flow data is much too large to process using existing models, and data-reduction methods are needed.


There has been some work on using reduced-basis models to compact data for emulation, including the functional linear models by \cite{Fea2006}, wavelet models by \cite{Bea2007} and principal component models by \cite{RS2002} and \cite{Hea2007}. Here, we employ a generalization of the latter method called \textit{proper orthogonal decomposition (POD)} \citep{Lum1967}, which is better known in statistical literature as the Karhunen-Lo\`eve decomposition \citep{Kar1947,Loe1955}. From a flow physics perspective, POD separates a simulated flow into key instability structures, each with its corresponding spatial and dynamic features. Such a decomposition is, however, inappropriate for emulation, because there is no way to connect the extracted instabilities of one input setting to the instabilities of another setting. To this end, we propose a new method called the \textit{common POD} (CPOD) to extract \textit{common} instabilities over the design space. This technique exploits a simple and physically justifiable linearity assumption on the spatial distribution of instability structures.

In addition to efficient flow emulation, our model also provides two important features. First, the same domain-specific model simplications (e.g., on the spatio-temporal correlation structure) which enable efficient prediction also allow for an efficient uncertainty quantification (UQ) for such a prediction. This UQ is highly valuable in practice, since the associated uncertainties for variable disturbance propagations can then be used for mitigating flow instabilities \citep{Yea2013}. Second, by incorporating known properties of the fluid flow into the model, the proposed emulator can in turn provide valuable insights on the dominant physics present in the system, which can then be used to guide further scientific investigations. One key example of this is the learning of dominant flow coupling mechanisms using a large co-kriging model \citep{SC1991,Bea2014} under sparsity constraints.

The paper is structured as follows. Section \ref{sec:data} provides a brief overview of the physical model of concern, including injector design, governing equations and experimental design. Section \ref{sec:method} introduces the proposed emulator model, and proposes a parallelized algorithm for efficient parameter estimation. Section \ref{sec:result} presents the emulation prediction and UQ for a new injector geometry, and interprets important physical correlations extracted by the emulator. Section 5 concludes with directions for future work.

\section{Injector schematic and large eddy simulations}\label{sec:data}

We first describe the design schematic for the swirl injector of concern, then briefly outline the governing partial differential equations and simulation tools. A discussion on experimental design is provided at the end of this section.

\begin{table}[t]
\begin{minipage}{0.53\textwidth}
\centering
\includegraphics[width=\textwidth]{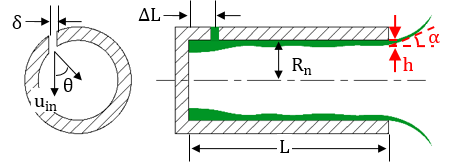}
\captionof{figure}{Schematic of injector configuration.}
\label{fig:inj}
\end{minipage}
\hfill
\begin{minipage}{0.43\textwidth}
\centering
\begin{tabular}{cc}
\toprule
\text{\bf{Parameter}} & \text{\bf{Range}}\\
\toprule
\text{$L$} & 20 mm - 100 mm\\
\text{$R_n$} & 2.0 mm - 5.0 mm\\
\text{$\delta$} & 0.5 mm - 2.0 mm\\
\text{$\theta$} & $45^{\circ} - 75^{\circ}$\\
\text{$\Delta L$} & 1.0 mm - 4.0 mm\\
\toprule
\end{tabular}
\caption{Range of geometric parameters.}
\label{tbl:range}
\end{minipage}
\end{table}

\subsection{Injector design}

Figure \ref{fig:inj} shows a schematic of the swirl injector under consideration. It consists of an open-ended cylinder and a row of tangential entries for liquid fluid injection. The configuration is typical of many propulsion and power-generation applications \citep{ZY2008, WY2016, Wea2017}. Liquid propellant is tangentially introduced into the injector and forms a thin film attached to the wall due to the swirl-induced centrifugal force. A low-density gaseous core exists in the center region in accordance with conservation of mass and angular momentum. The liquid film exits the injector as a thin sheet and mixes with the ambient gas. The swirl injection and atomization process involves two primary mechanisms: disintegration of the liquid sheet as it swirls and stretches, and sheet breakup due to the interaction with the surroundings. The design of the injector significantly affects the atomization characteristics and stability behaviors.

{Figure} \ref{fig:inj} shows the five design variables considered for injector geometry: the injector length $L$, the nozzle radius $R_n$, the inlet diameter $\delta$, the injection angle $\theta$, and the distance between inlet and head-end $\Delta L$. From flow physics, these five variables are influential for liquid film thickness $h$ and spreading angle $\alpha$ (see Figure \ref{fig:inj}), which are key measures of injector performance of a swirl injector. For example, a larger injection angle $\theta$ induces greater swirl momentum in the liquid oxygen flow, which in turn causes thinner film thickness and smaller spreading angle. Table \ref{tbl:range} summarizes the design ranges for these five variables. To ensure the applicability of our work, broad geometric ranges are considered, covering design settings for several existing rocket injectors. Specifically, the range for injector length $L$ covers the length of RD-0110 and RD-170 liquid-fuel rocket engines.

\subsection{Flow simulation}

The numerical simulations here are performed with a pressure of 100 atm, which is typical of contemporary liquid rocket engines with liquid oxygen (LOX) as the propellant. The physical processes modeled here are turbulent flows, in which various sizes of turbulent eddies are involved. A direct numerical simulation to resolve all eddy length-scales is computationally prohibitive. To this end, we employ the large eddy simulation (LES) technique, which directly simulates large turbulent eddies and employs a model-based approach for small eddies. To provide initial turbulence, broadband Gaussian noise is superimposed onto the inlet velocity components. Thermodynamic and transport properties are simulated using the techniques in \cite{Hea2014} and \cite{Wea2015}; the theoretical and numerical framework can be found in \cite{OY1998} and \cite{Zea2004}. To optimize computational speed, a multi-block domain decomposition technique combined with the message-passing interface for parallel computing is applied. Each LES simulation takes 6 days of computation time, parallelized over 200 CPU cores, to obtain $T = 1,000$ snapshots with a time-step of 0.03 ms after the flow reaches statistically stationary state. From this, six flow variables of interest can be extracted: axial ($u$), radial ($v$), and circumferential ($w$) components of velocity, temperature ($T$), pressure ($P$) and density ($\rho$).

Numerical simulations are conducted for $n=30$ injector geometries in the timeframe set for this project. These simulation runs are allocated over the design space in Table \ref{tbl:range} using the maximum projection (MaxPro) design proposed by \cite{Jea2015}. Compared to Latin-hypercube-based designs (e.g., \citealp{Mea1979}, \citealp{MM1995}), MaxPro designs enjoy better space-filling properties in {all} possible projections of the design space, and also provide better predictions for GP modeling. While $n=30$ simulation runs may appear to be too small of a dataset for training the proposed flow emulator, we show this sample size can provide accurate flow predictions for the application at hand, through an elicitation of flow physics and the incorporation of such physics into the model. For these 30 runs, one issue which arises is that the simulation data is massive, requiring nearly a hundred gigabytes in computer storage. For such large data, a blind application of existing flow kriging methods may require weeks for flow prediction, which entirely defeats the purpose of emulation, because simulated flows can generated in 6 days. Again, by properly eliciting and incorporating physics as simplifying assumptions for the emulator model, accurate flow predictions can be achieved in hours despite a limited run size. We elaborate on this elicitation procedure in the following section.



\section{Emulator model}\label{sec:method}

\begin{table}[t]
\centering
\begin{tabular}{K{0.6\linewidth} | K{0.38\linewidth}}
\toprule
\textbf{Flow physics} & \textbf{Model assumption}\\
\toprule
Coherent structures in turbulent flow \citep{Lum1967} & POD-based kriging\\
\hline
Similar Reynolds numbers for cold-flows \citep{Sto1851} & Linear-scaling modes in CPOD \\
\hline
Dense simulation time-steps & Time-independent emulator\\
\hline
Couplings between flow variables \citep{Pop2001} & Co-kriging framework with covariance matrix $\bm{T}$\\
\hline
Few-but-significant couplings \citep{Pop2001} & Sparsity on $\bm{T}^{-1}$\\
\toprule
\end{tabular}
\caption{Elicited flow physics and corresponding assumptions for the emulator model.}
\label{tbl:model}
\end{table}

We first introduce the new idea of CPOD, then present the proposed emulator model and a parallelized algorithm for parameter estimation. A key theme in this section (and indeed, for this paper) is the elicitation and incorporation of flow physics within the emulator model. This not only allows for {efficient} and {accurate} flow predictions through simplifying model assumptions, but also provides a data-driven method for {extracting} useful flow physics, which can then guide future experiments. As demonstrated in Section 4, both objectives can be achieved despite limited runs and complexities inherent in flow data. Table \ref{tbl:model} summarizes the elicited flow physics and the corresponding emulator assumptions; we discuss each point in greater detail below.

\subsection{Common POD}
A brief overview of POD is first provided, following \cite{Lum1967}. For a \textit{fixed} injector geometry, let $Y(\bm{x},t)$ denote a flow variable (e.g., pressure) at spatial coordinate $\bm{x}\in \mathbb{R}^2$ and flow time $t$. POD provides the following decomposition of $Y(\bm{x},t)$ into separable spatial and temporal components:
\begin{equation}
Y(\bm{x},t) =\sum_{k=1}^\infty \beta_k(t) \phi_k(\bm{x}), 
\label{eq:klexp}
\end{equation}
with the spatial eigenfunctions $\{\phi_k(\bm{x})\}_{k=1}^\infty$ and temporal coefficients $\{\beta_k(t)\}_{k=1}^\infty$ given by:
\begin{align}
\begin{split}
\phi_k(\bm{x}) = \argmax_{\substack{\| \psi \|_2 = 1, \\ \langle \psi, \phi_l \rangle = 0, \forall l < k}} \int \left\{ \int Y(\bm{x},t) \psi(\bm{x}) \; d\bm{x}\right\}^2 \; dt, \quad
\beta_k(t) = \int Y(\bm{x},t) \phi_k(\bm{x}) \; d\bm{x}.
\end{split}
\end{align}
Following \cite{Bea1993}, we refer to $\{\phi_k(\bm{x})\}_{k=1}^\infty$ as the \textit{spatial POD modes} for $Y(\bm{x},t)$, and its corresponding coefficients $\{\beta_k(t)\}_{k=1}^\infty$ as \textit{time-varying coefficients}.

There are two key reasons for choosing POD over other reduced-basis models. First, one can show \citep{Loe1955} that any truncated representation in \eqref{eq:klexp} gives the best flow reconstruction of $Y(\bm{x},t)$ in $L_2$-norm, compared to any other linear expansion of space/time products with the same number of terms. This property is crucial for our application, since it allows the massive simulation data to be optimally reduced to a smaller training dataset for the proposed emulator. Second, the POD has a special interpretation in terms of turbulent flow. In the seminal paper by \cite{Lum1967}, it is shown that, under certain conditions, the expansion in \eqref{eq:klexp} can extract \textit{physically meaningful} coherent structures which govern turbulence instabilities. For this reason, physicists use POD as an experimental tool to pinpoint key flow instabilities, simply through an inspection of $\phi_k(\bm{x})$ and the dominant frequencies in $\beta_k(t)$. For example, using POD analysis, \cite{ZY2008} showed that the two flow phenomena, hydrodynamic wave propagation on LOX film and vortex core excitation near the injector exit, are the key mechanisms driving flow instability. This is akin to the use of {principal components} in regression, which can yield meaningful results in applications where such components have innate interpretability.

\begin{figure}[t]
\centering
\includegraphics[width=\textwidth]{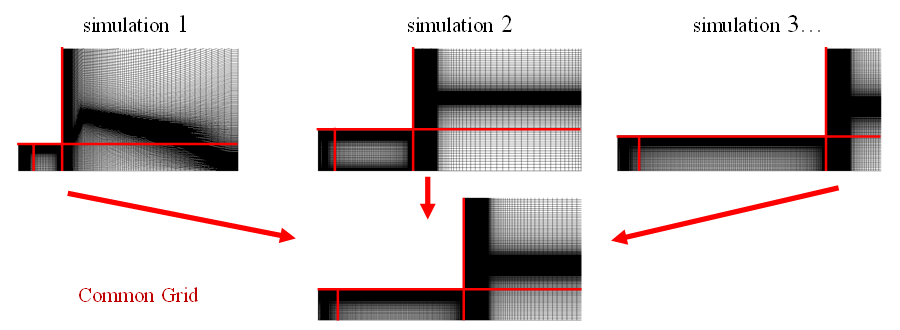}
\caption{Common grid using linearity assumption for CPOD.}
\label{fig:rescalestep1}
\end{figure}

Unfortunately, POD is only suitable for extracting instability structures at a \textit{single} geometry, whereas for emulation, a method is needed that can extract {common} structures over \textit{varying} geometries. With this in mind, we propose a new decomposition called common POD (CPOD). The key assumption of CPOD is that, under a \textit{physics-guided partition} of the computational domain, the spatial distribution of coherent structures \textit{scales linearly} over varying injector geometries. For {cold flows}, this can be justified by similar Reynolds numbers (which characterize flow dynamics) over different geometries \citep{Sto1851}. This is one instance of model simplification through elicitation, because such a property likely does not hold for general flows. This linearity assumption is highly valuable for computational efficiency, because flows from different geometries can then be rescaled onto a common spatial grid for instability extraction. Figure \ref{fig:rescalestep1} visualizes this procedure. The grids for each simulation are first split into four parts: from injector head-end to the inlet, from the inlet to the nozzle exit, and the top and bottom portions of the downstream region. Each part is then proportionally rescaled to a common, reference grid according to changes in the geometric variables $L$, $R_n$ and $\Delta L$ (see Figure \ref{fig:inj}). From a physics perspective, such a partition is necessary for the linearity assumption to hold.

Stating this mathematically, let $\bm{c}_1, \cdots, \bm{c}_n \in \mathbb{R}^p$ be the $n$ simulated geometries, let $Y(\bm{x},t;\bm{c}_i)$ be the simulated flow at setting $\bm{c}_i$, and fix some setting $\bm{c} \in \{\bm{c}_i\}_{i=1}^n$ as the geometry for the common grid. Next, define $\mathcal{M}_i:\mathbb{R}^2 \rightarrow \mathbb{R}^2$ as the linear map which rescales spatial modes on the common geometry $\bm{c}$ back to the $i$-th simulated geometry $\bm{c}_i$ according to geometric changes in $L$, $R_n$ and $\Delta L$. $\mathcal{M}_i$ can be viewed as the inverse map of the procedure described in the previous paragraph and visualized in Figure \ref{fig:rescalestep1}, which rescales modes from $\bm{c}_i$ to the common geometry $\bm{c}$ (see Appendix A.1 for details). CPOD provides the following decomposition of  $Y(\bm{x},t; \bm{c}_i)$:
\begin{equation}
Y(\bm{x},t; \bm{c}_i) = \sum_{k=1}^\infty \beta_k(t; \bm{c}_i) \mathcal{M}_i \{\phi_k(\bm{x})\},
\label{eq:cklexp}
\end{equation}
with the spatial CPOD modes $\{\phi_k(\bm{x})\}$ and time-varying coefficients $\{\beta_k(t;\bm{c}_i)\}$ defined as:
\small
\begin{equation}
\phi_k(\bm{x}) = \argmax_{\substack{\| \psi \|_2 = 1, \\ \langle \psi, \phi_l \rangle = 0, \forall l < k}} \sum_{i=1}^n \int \left\{ \int Y(\bm{x},t; \bm{c}_i) \mathcal{M}_i\{\psi(\bm{x})\} \; d\bm{x} \right\}^2 dt, \; \beta_k(t;\bm{c}_i) = \int Y(\bm{x},t; \bm{c}_i) \mathcal{M}_i \{ \phi_k(\bm{x}) \} \; d\bm{x}.
\end{equation}
\normalsize
Here, $\phi_k(\bm{x})$ is the spatial distribution for the $k$-th common flow structure, with $\beta_k(t;\bm{c}_i)$ its time-varying coefficient for geometry $\bm{c}_i$. As in POD, leading terms in CPOD can also be interpreted in terms of flow physics, a property we demonstrate later in Section \ref{sec:result}. CPOD therefore not only provides optimal {data-reduction} for the simulation data, but also extracts {physically meaningful} structures which can then be incorporated for emulation.

Algorithmically, the CPOD expansion can be computed by rescaling and interpolating all flow simulations to the common grid, computing the POD expansion, and then rescaling the resulting modes back to their original grids. Interpolation is performed using the inverse distance weighting method in \cite{She1968}, and can be justified by dense spatial resolution of the data (with around 100,000 grid points for each simulation). Letting $T$ be the total number of time-steps, a naive implementation of this decomposition requires $O(n^3T^3)$ work, due to a singular-value-decomposition (SVD) step. Such a decomposition therefore becomes computationally intractable when the number of runs grows large or when simulations have dense time-steps (as is the case here). To avoid this computational issue, we use an iterative technique from \cite{LS1998} called {the implicitly restarted Arnoldi method}, which approximates leading terms in \eqref{eq:cklexp} using periodically restarted Arnoldi decompositions. The full algorithm for CPOD is outlined in Appendix A.


\subsection{Model specification}
After the CPOD extraction, the extracted time-varying coefficients $\{\beta_k(t;\bm{c}_i)\}_{i,k}$ are then used as data for fitting the proposed emulator. There has been some existing work on dynamic emulator models, such as \cite{CO2010}, \cite{Cea2009} and \cite{LW2009}, but the sheer number of simulation time-steps here can impose high computation times and numerical instabilities for these existing methods \citep{Hea2015}. As mentioned previously, computational efficiency is paramount for our problem, since simulation runs can be performed within a week. Moreover, existing emulators cannot account for cross-correlations between different dynamic systems, while the flow physics represented by different CPOD modes are known to be highly coupled from governing equations. Here, we exploit the dense temporal resolution of the flow by using a \textit{time-independent (TI)} emulator that employs independent kriging models at {each slice of time}. The rationale is that, because time-scales are so fine, there is no practical need to estimate temporal correlations (even when they exist), since prediction is not required between time-steps. This time-independent simplification is key for emulator efficiency, since it allows us to fully exploit the power of parallel computing for model fitting and flow prediction.

The model is as follows. Suppose $R$ flow variables are considered (with $R=6$ in the present case), and the CPOD expansion in \eqref{eq:cklexp} is truncated at $K_r$ terms for flow $r = 1, \cdots, R$. Let $\boldsymbol{\beta}^{(r)}(t;\bm{c}) = (\beta^{(r)}_1(t;\bm{c}), \cdots, \beta^{(r)}_{K_r}(t;\bm{c}))^T$ be the vector of $K_r$ time-varying coefficients for flow variable $r$ at design setting $\bm{c}$, with $\boldsymbol{\beta}(t;\bm{c}) = (\boldsymbol{\beta}^{(1)}(t;\bm{c})^T, \cdots, \boldsymbol{\beta}^{(R)}(t;\bm{c})^T)^T$ the coefficient vector for all flows at $\bm{c}$. We assume the following \textit{time-independent GP model} on $\boldsymbol{\beta}(t;\bm{c})$:
\begin{equation}
\boldsymbol{\beta}(t;\bm{c}) \sim GP\{\boldsymbol{\mu}(t), \boldsymbol{\Sigma}(\cdot, \cdot;t)\}, \quad \boldsymbol{\beta}(t;\bm{c}) \perp \boldsymbol{\beta}(t';\bm{c}) \text{ for } t \neq t'.
\label{eq:gpcoef}
\end{equation}
Here, $K = \sum_{r=1}^R K_r$ is the number of extracted modes over all $R$ flow variables, $\boldsymbol{\mu}\in \mathbb{R}^K$ is the process mean vector, and $\boldsymbol{\Sigma}(\cdot, \cdot): \mathbb{R}^p \times \mathbb{R}^p \rightarrow \mathbb{R}^{K \times K}$ its corresponding covariance matrix function defined below. Since the GPs are now time-independent, we present the specification for \textit{fixed} time $t$, and refer to $\boldsymbol{\beta}(t;\bm{c})$, $\boldsymbol{\mu}(t)$ and $\bm{\Sigma}(\cdot,\cdot;t)$ as $\boldsymbol{\beta}(\bm{c})$, $\boldsymbol{\mu}$ and $\bm{\Sigma}(\cdot,\cdot)$ for brevity.

For computational efficiency, the following separable form is assumed for $\boldsymbol{\Sigma}(\cdot,\cdot)$:
\begin{equation}
\boldsymbol{\Sigma}(\bm{c}_1, \bm{c}_2) = r_\tau(\bm{c}_1, \bm{c}_2) \bm{T}, \quad r_\tau(\bm{c}_1, \bm{c}_2) = \prod_{j=1}^p \tau_j^{4(c_{1j} - c_{2j})^2}, \quad \bm{c}_1, \bm{c}_2 \in \mathbb{R}^p, \quad \tau_j \in (0,1),
\label{eq:gpcov}
\end{equation}
where $\bm{T} \in \mathbb{R}^{K \times K}$ is a symmetric, positive definite matrix called the \textit{CPOD covariance matrix}, and $r_\tau(\cdot,\cdot)$ is the correlation function over the design space, parameterized by $\boldsymbol{\tau} = (\tau_1, \cdots, \tau_p)^T \in (0,1)^p$. This can be viewed as a large co-kriging model \citep{SC1991} over the design space, with the multivariate observations being the extracted CPOD coefficients for all flow variables. Note that $r_{\tau}$ is a reparametrization of the squared-exponential (or Gaussian) correlation function $\exp\{-\sum_{j=1}^p \theta_j (c_{1j}-c_{2j})^2\}$, with $\theta_j = -4 \log \tau_j$. In our experience, such a reparametrization allows for a more numerically stable optimization of MLEs, because the optimization domain $\tau_j \in (0,1)$ is now bounded. Our choice of the Gaussian correlation is also well-justified for the application at hand, since fully-developed turbulence dynamics are known to be relatively smooth.

Suppose simulations are run at settings $\bm{c}_1, \cdots, \bm{c}_n$, and assume for now that model parameters are known. Invoking the conditional distribution of the multivariate normal distribution, the time-varying coefficients at a new setting $\bm{c}_{new}$ follow the distribution:
\small
\begin{align}
\begin{split}
\boldsymbol{\beta}(\bm{c}_{new})|\{\boldsymbol{\beta}(\bm{c}_i)\}_{i=1}^n \sim  \mathcal{N} \Bigg( &\boldsymbol{\mu} + 
\left( \bm{T} \otimes \bm{r}_{\tau,new}\right)^T
\left( \bm{T}^{-1} \otimes \bm{R}_{\tau}^{-1} \right)
\left(\boldsymbol{\beta} - 
\bm{1}_n
\otimes
\boldsymbol{\mu}\right), \\
& \quad \bm{T} - \left( \bm{T} \otimes \bm{r}_{\tau,new}\right)^T \left( \bm{T}^{-1} \otimes \bm{R}_{\tau}^{-1} \right) \left( \bm{T} \otimes \bm{r}_{\tau,new} \right) \Bigg),
\label{eq:coefdist}
\end{split}
\end{align}
\normalsize
where $\bm{r}_{\tau,new} = (r_\tau(\bm{c}_{new},\bm{c}_1), \cdots, r_\tau(\bm{c}_{new},\bm{c}_n))^T$ and $\bm{R}_\tau= {[r_\tau(\bm{c}_i,\bm{c}_j)]^n_{i=1}}_{j=1}^n$. Using algebraic manipulations, the minimum-MSE (MMSE) predictor for $\boldsymbol{\beta}(\bm{c}_{new})|\{\boldsymbol{\beta}(\bm{c}_i)\}_{i=1}^n$ and its corresponding variance is given by
\small 
\begin{equation}
\boldsymbol{\hat{\beta}}(\bm{c}_{new}) = \boldsymbol{\mu} +\left( (\bm{r}^T_{\tau,new}\bm{R}_{\tau}^{-1})\otimes\bm{I}_{K}\right)
\left(\boldsymbol{\beta} - 
\bm{1}_n
\otimes
\boldsymbol{\mu} \right),\mathbb{V}\{\boldsymbol{\beta}(\bm{c}_{new})|\{\boldsymbol{\beta}(\bm{c}_i)\}^n_{i=1}\} = \left(1 - \bm{r}_{\tau,new}^T \bm{R}_{\tau}^{-1} \bm{r}_{\tau,new} \right) \bm{T},
\label{eq:coefpred}
\end{equation}
\normalsize
where $\bm{I}_K$ and $\bm{1}_n$ denote a $K \times K$ identity matrix and a 1-vector of $n$ elements, respectively. Substituting this into the CPOD expansion \eqref{eq:cklexp}, the predicted $r$-th flow variable becomes:
\begin{equation}
\hat{Y}^{(r)}(\bm{x},t; \bm{c}_{new}) = \sum_{k=1}^{K_r} \hat{\beta}^{(r)}_k(\bm{c}_{new}) \mathcal{M}_{new} \{\phi^{(r)}_k(\bm{x})\},
\label{eq:flowpred}
\end{equation}
with the associated spatio-temporal variance:
\small
\begin{equation}
\mathbb{V} \{{Y}^{(r)}(\bm{x},t; \bm{c}_{new})| \{{Y}^{(r)}(\bm{x},t; \bm{c}_{i})\}_{i=1}^n\} = \sum_{k=1}^{K_r} \mathbb{V}\{{\beta}^{(r)}_k(\bm{c}_{new}) | \{\boldsymbol{\beta}(\bm{c}_i)\}^n_{i=1}\} \} \left[\mathcal{M}_{new} \{\phi^{(r)}_k(\bm{x})\} \right]^2,
\label{eq:flowvar}
\end{equation}
\normalsize
where $\phi^{(r)}_k(\bm{x})$ is the $k$-th CPOD mode for flow variable $r$. This holds because the CPOD modes for a {fixed} flow variable are orthogonal (see Section 3.1).

It is worth noting that, when model parameters are {known}, the MMSE predictor in \eqref{eq:coefpred} from the proposed co-kriging model (which we call $M_A$) is the same as the MMSE predictor from the simpler \textit{independent} GP model with $\bm{T}$ diagonal (which we call $M_0$). One advantage of the co-kriging model $M_A$, however, is that it provides improved UQ compared to the independent model $M_0$, as we show below. Moreover, the MMSE predictor for a derived function $g$ of the flow can be quite different between $M_A$ and $M_0$. This is demonstrated in the study of turbulent kinetic energy in Section 4.3.

\subsubsection{CPOD covariance matrix}
\label{sec:covmat}

\begin{figure}[t]
\centering
\includegraphics[width=0.5\textwidth]{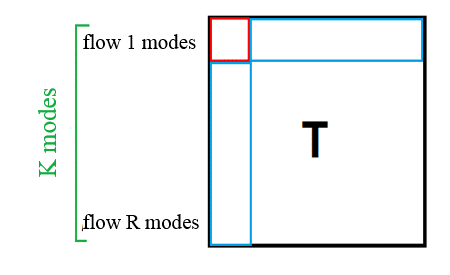}
\caption{Illustration of the CPOD correlation matrix $\bm{T}$. Red indicates a diagonal matrix, while blue indicates non-diagonal entries.}
\label{fig:TT}
\end{figure}

We briefly describe why the CPOD covariance matrix $\bm{T}$ is appealing from both a physical and a statistical perspective. From the underlying governing equations, it is well known that certain dynamic behaviors are strongly \textit{coupled} for different flow variables \citep{Pop2001}. For example, pressure oscillation in the form of acoustic waves within an injector can induce velocity and density fluctuations. In this sense, $\bm{T}$ incorporates knowledge of these physical couplings within the {emulator} itself, with $\bm{T}_{ij}\gg 0$ indicating the presence of a significant coupling between modes $i$ and $j$, and vice versa. The covariance selection and estimation of $\bm{T}$ therefore provide a data-driven way to \textit{extract} and \textit{rank} significant flow couplings, which is of interest in itself and can be used to guide further experiments. Note that the block submatrices of $\bm{T}$ corresponding to the {same} flow variable (marked in red in Figure \ref{fig:TT}) should be diagonal, by the orthogonality of CPOD modes.

The CPOD covariance matrix $\bm{T}$ also plays an important {statistical} role in emulation. Specifically, when significant cross-correlations exist between modes (which we know to be true from the flow couplings imposed by governing equations), the incorporation of this correlation structure within our model ought to provide a more {accurate} quantification of uncertainty. This is indeed true, and is made precise by the following theorem.

\begin{theorem}
Consider the two models $M_0: \boldsymbol{\beta}(\bm{c}) \in \mathbb{R}^K \sim GP\{\boldsymbol{\mu}, \boldsymbol{\Sigma}^{(0)}\}$ and $M_A: \boldsymbol{\beta}(\bm{c}) \sim GP\{\boldsymbol{\mu}, \boldsymbol{\Sigma}^{(A)}\}$, where $\boldsymbol{\Sigma}^{(0)}(\bm{c}_1, \bm{c}_2)= r_\tau(\bm{c}_1, \bm{c}_2) \bm{D}$ and $\boldsymbol{\Sigma}^{(A)}(\bm{c}_1, \bm{c}_2)= r_\tau(\bm{c}_1, \bm{c}_2) \bm{T}$ with $\bm{T} \succeq 0$ and $\bm{D} = \textup{diag}\{\bm{T}\}$. Let $C_0$ be the $100(1-\alpha)\%$ highest-density confidence region (HDCR, see \citealp{Hyn1996}) of $\boldsymbol{\beta}(\bm{c}_{new})|\{\boldsymbol{\beta}(\bm{c}_i)\}_{i=1}^n$ under $M_0$. Suppose $\lambda_{min}(\bm{T}^{1/2}\bm{D}^{-1}\bm{T}^{1/2}) > 1$. Then:
\[\mathbb{P}\left\{\boldsymbol{\beta}(\bm{c}_{new}) \in C_0|M_A, \{\boldsymbol{\beta}(\bm{c}_i)\}_{i=1}^n \right\} < 1-\alpha.\]
\label{thm:uq}
\end{theorem}

\begin{proof}
For brevity, let $\boldsymbol{\beta} \equiv \boldsymbol{\beta}(\bm{c}_{new})|\{\boldsymbol{\beta}(\bm{c}_i)\}_{i=1}^n$, and let $\hat{\boldsymbol{\beta}} \equiv \mathbb{E}[\boldsymbol{\beta}(\bm{c}_{new})|\{\boldsymbol{\beta}(\bm{c}_i)\}_{i=1}^n]$. Letting $\bm{Z} \sim \mathcal{N}(\bm{0},\bm{I}_K)$, it is easy to show that
\[\boldsymbol{\beta} - \hat{\boldsymbol{\beta}} | M_0 \sim \mathcal{N}\left\{\bm{0}, \left(1 - \bm{r}_{\tau,new}^T \bm{R}_{\tau}^{-1} \bm{r}_{\tau,new} \right) \bm{D} \right\} \stackrel{d}{=} \sqrt{1 - \bm{r}_{\tau,new}^T \bm{R}_{\tau}^{-1} \bm{r}_{\tau,new} } \bm{D}^{1/2} \bm{Z}, \quad \text{and}\]
\[\boldsymbol{\beta} - \hat{\boldsymbol{\beta}} | M_A \sim \mathcal{N}\left\{\bm{0}, \left(1 - \bm{r}_{\tau,new}^T \bm{R}_{\tau}^{-1} \bm{r}_{\tau,new} \right) \bm{T} \right\} \stackrel{d}{=} \sqrt{1 - \bm{r}_{\tau,new}^T \bm{R}_{\tau}^{-1} \bm{r}_{\tau,new} } \bm{T}^{1/2} \bm{Z}.\]
Under the independent model $M_0$, the $100(1-\alpha)\%$ HDCR becomes:
\[C_0 = \{\boldsymbol{\xi} \; : \; \left(1 - \bm{r}_{\tau,new}^T \bm{R}_{\tau}^{-1} \bm{r}_{\tau,new} \right)^{-1} (\boldsymbol{\xi} - \hat{\boldsymbol{\beta}})^T \bm{D}^{-1} (\boldsymbol{\xi} - \hat{\boldsymbol{\beta}}) \leq \chi^2_K(1-\alpha) \},\]
where $\chi^2_K(1-\alpha)$ be the $(1-\alpha)$-quantile of a $\chi^2$-distribution with $K$ degrees of freedom. Now, let $\lambda_{min}$ denote the minimum eigenvalue of $\bm{T}^{1/2} \bm{D}^{-1} \bm{T}^{1/2}$. It follows that
\begin{align*}
\mathbb{P}\left(\boldsymbol{\beta} \in C_0|M_A\right) &= \mathbb{P}\left\{ (\boldsymbol{\beta} - \hat{\boldsymbol{\beta}})^T \bm{D}^{-1} (\boldsymbol{\beta} - \hat{\boldsymbol{\beta}}) \leq \left(1 - \bm{r}_{\tau,new}^T \bm{R}_{\tau}^{-1} \bm{r}_{\tau,new} \right) \chi^2_K(1-\alpha) \Big| M_A \right\}\\
&= \mathbb{P}\left\{ \bm{Z}^T (\bm{T}^{1/2} \bm{D}^{-1} \bm{T}^{1/2}) \bm{Z} \leq \chi^2_K(1-\alpha) \right\}\\
&\leq \mathbb{P}\left\{ \bm{Z}^T \bm{Z} \leq \lambda_{min}^{-1} \chi^2_K(1-\alpha) \right\},
\end{align*}
since $\bm{Z}^T (\bm{T}^{1/2} \bm{D}^{-1} \bm{T}^{1/2}) \bm{Z} \geq \lambda_{min}\bm{Z}^T \bm{Z}$ almost surely. The asserted result follows because $\mathbb{P}\left\{ \bm{Z}^T \bm{Z} \leq \lambda_{min}^{-1} \chi^2_K(1-\alpha) \right\}$ is strictly less than $1-\alpha$ when $\lambda_{min} > 1$.

\end{proof}

In words, this theorem quantifies the effect on coverage probability when the true co-kriging model $M_A$, which accounts for cross-correlations between modes, is misspecified as $M_0$, the independent model ignoring such cross-correlations. Note that an increase in the number of significant {non-zero cross-correlations} in $\bm{T}$ causes $\bm{T}^{1/2} \bm{D}^{-1} \bm{T}^{1/2}$ to deviate further from unity, which in turn may increase $\lambda_{min}$. Given enough such correlations, Theorem \ref{thm:uq} shows that the coverage probability from the misspecified model $M_0$ is less than the desired $100(1-\alpha)\%$ rate. In the present case, this suggests that when there are enough significant {flow couplings}, the co-kriging model $M_A$ provides more {accurate} UQ for the \textit{joint} prediction of flow variables when compared to the misspecified, independent model $M_0$. This improvement also holds for functions of flow variables (as we demonstrate later in Section \ref{sec:result}), although a formal argument is not presented here.

It is important to mention here an important trade-off for co-kriging models in general, and why the proposed model is appropriate for the application at hand in view of such a trade-off. It is known from spatial statistics literature (see, e.g., \citealp{Bea2014, Mea2016}) that when the matrix $\bm{T}$ exhibits strong correlations and can be estimated well, one enjoys improved predictive performance through a co-kriging model (this is formally shown for the current model in Theorem \ref{thm:uq}). However, when such correlations are absent or cannot be estimated well, a co-kriging model can yield poorer performance to an independent model! We claim that the former is true for the current application at hand. First, the differential equations governing the simulation procedure explicitly impose strong dependencies between flow variables, so we know \textit{a priori} the existence of strong correlations in $\bm{T}$. Second, we will show later in Section \ref{sec:correx} that the dominant correlations selected in $\bm{T}$ are physically interpretable in terms of fluid mechanic principles and conservation laws, which provides strong evidence for the correct estimation of $\bm{T}$.

One issue with fitting $M_A$ is that there are many more parameters to estimate. Specifically, since the CPOD covariance matrix $\bm{T}$ is $K \times K$ dimensional, there is {insufficient} data for estimating all entries in $\bm{T}$ using the extracted coefficients from the CPOD expansion. One solution is to impose the sparsity constraint $\|\bm{T}^{-1}\|_1 \leq \gamma$, where $\|\bm{A}\|_1 = \sum_{k=1}^K \sum_{l=1}^K |A_{kl}|$ is the element-wise $L_1$ norm. For a small choice of $\gamma$, this forces nearly all entries in $\bm{T}^{-1}$ to be zero, thus permitting consistent estimation of the {few significant} correlations. Sparsity can also be justified from an engineering perspective, because the number of significant couplings is known to be small from flow physics. $\gamma$ can also be adjusted to extract a {pre-specified} number of flow couplings, which is appealing from an engineering point-of-view. The justification for sparsifying $\bm{T}^{-1}$ instead of $\bm{T}$ is largely computational, because, algorithmically, the former problem can be handled much more efficiently than the latter using the graphical LASSO (\citealp{Fea2008}; see also \citealp{BT2011}). Such efficiency is crucial here, since GP parameters need to be jointly estimated as well.

Although the proposed model is {similar} to the one developed in \cite{Qea2008} for emulating qualitative factors, there are two key distinctions. First, our model allows for {different} process variances for each coefficient, whereas their approach restricts all coefficients to have {equal} variances. Second, our model incorporates sparsity on the CPOD covariance matrix, an assumption necessary from a {statistical} point-of-view and appealing from a physics extraction perspective. Lastly, the algorithm proposed below can estimate $\bm{T}$ more efficiently than the semi-definite programming approach in \cite{Qea2008}.

\subsection{Parameter estimation}
To estimate the model parameters $\boldsymbol{\mu}$, $\bm{T}$ and $\boldsymbol{\tau}$, maximum-likelihood estimation (MLE) is used in favor of a Bayesian implementation. The primary reason for this choice is computational efficiency: for the proposed emulator to be used as a fast investigative tool for surveying the design space, it should generate flow predictions much quicker than a direct LES simuation, which requires several days of parallelized computation.


From \eqref{eq:gpcoef} and \eqref{eq:gpcov}, the maximum-likelihood formulation can be written as $\argmin_{\boldsymbol{\mu}, \bm{T}, \boldsymbol{\tau}} \allowbreak l_\lambda(\boldsymbol{\mu}, \bm{T}, \boldsymbol{\tau})$, where $l_\lambda(\boldsymbol{\mu}, \bm{T}, \boldsymbol{\tau})$ is the \textit{penalized} negative log-likelihood:
\small
\begin{equation}
l_\lambda(\boldsymbol{\mu}, \bm{T}, \boldsymbol{\tau}) = n\log \det\bm{T} + K\log \det \bm{R}_\tau+(\bm{B}-\bm{1}_n \otimes \boldsymbol{\mu})^T [ \bm{R}_\tau^{-1}\otimes{\bm{T}}^{-1}](\bm{B}-\bm{1}_n \otimes \boldsymbol{\mu}) + \lambda \|\bm{T}^{-1}\|_1.
\label{eq:nll}
\end{equation}
\normalsize
Note that, because the formulation is convex in $\bm{T}^{-1}$, the sparsity constraint $\|\bm{T}^{-1}\|_1 \leq \gamma$ has been incorporated into the likelihood through the penalty $\lambda \|\bm{T}^{-1}\|_1$ using strong duality. Similar to $\gamma$, a {larger} $\lambda$ results in a smaller number of selected correlations, and vice versa. The tuning method for $\lambda$ should depend on the desired end-goal. For example, if predictive accuracy is the primary goal, then $\lambda$ should be tuned using cross-validation techniques \citep{Fea2001}. However, if correlation extraction is desired or prior information is available on flow couplings, then $\lambda$ should be set so that a {fixed} (preset) number of correlations is extracted. We discuss this further in Section \ref{sec:result}.

\begin{algorithm}[t]
\caption{BCD algorithm for maximum likelihood estimation}
\label{alg:mle}
\begin{algorithmic}[1]
\small
\ParFor{each time-step $t = 1, \cdots, T$}
\stb Set initial values $\boldsymbol{\mu} \leftarrow \bm{0}_K$, $\bm{T} \leftarrow \bm{I}_{K}$ and $\boldsymbol{\tau} \leftarrow \bm{1}_p$, and set $\bm{B} \leftarrow (\boldsymbol{\beta}(\bm{c}_1), \cdots, \boldsymbol{\beta}(\bm{c}_n))^T$
\Repeat\\
\quad \quad \quad \underline{Optimizing $\bm{T}$}:
\stb Set $\bm{W} \leftarrow \frac{1}{n}{(\bm{B} - \bm{1}_n \otimes \boldsymbol{\mu}^T )^T \bm{R}_\tau^{-1}(\bm{B} - \bm{1}_n \otimes \boldsymbol{\mu}^T )} + \lambda \cdot \bm{I}_{K}$
\Repeat
\For{$j = 1, \cdots, K$}
\stb Solve $\tilde{\boldsymbol{\delta}} = \argmin_{\boldsymbol{\delta}} \left\{ \frac{1}{2} \|\bm{W}_{-j, -j}^{1/2}\boldsymbol{\delta}\|_2^2 + \lambda \|\boldsymbol{\delta}\|_1 \right\}$ using LASSO
\stb Update $\bm{W}_{-j,j} \leftarrow \bm{W}_{-j,-j} \tilde{\boldsymbol{\delta}}$ and $\bm{W}_{j,-j}^T \leftarrow \bm{W}_{-j,-j} \tilde{\boldsymbol{\delta}}$
\EndFor
\Until{$\bm{W}$ converges}
\stb Update $\bm{T} \leftarrow \bm{W}^{-1}$\\
\quad \quad \quad \underline{Optimizing $\boldsymbol{\mu}$ and $\boldsymbol{\tau}$}:
\stb Update $\boldsymbol{\tau} \leftarrow \argmin_{\tau} l_\lambda(\boldsymbol{\mu}_{\boldsymbol{\tau}}, \bm{T}, \boldsymbol{\tau})$ with L-BFGS, with $\boldsymbol{\mu}_{\boldsymbol{{\tau}}} = (\bm{1}_n^T \bm{R}_{{\tau}}^{-1}\bm{1}_n)^{-1} (\bm{1}_n^T \bm{R}_{{\tau}}^{-1} \bm{B})$ 
\stb Update $\boldsymbol{\mu} \leftarrow \boldsymbol{\mu}_{\boldsymbol{\tau}}$
\Until{$\boldsymbol{\mu}$, $\bm{T}$ and $\boldsymbol{\tau}$ converge}
\EndParFor
\stb \Return $\boldsymbol{\mu}(t)$, $\bm{T}(t)$ and $\boldsymbol{\tau}(t)$
\normalsize
\end{algorithmic}
\end{algorithm}

Assume for now a fixed penalty $\lambda>0$. To compute the MLEs in \eqref{eq:nll}, we propose the following \textit{blockwise coordinate descent} (BCD) algorithm. First, assign initial values for $\boldsymbol{\mu}$, $\bm{T}$ and $\boldsymbol{\tau}$. Next, iterate the following two updates until parameters converge: (a) for fixed GP parameters $\boldsymbol{\mu}$ and $\boldsymbol{\tau}$, optimize for $\bm{T}$ in \eqref{eq:nll}; and (b) for fixed covariance matrix $\bm{T}$, optimize for $\boldsymbol{\mu}$ and $\boldsymbol{\tau}$ in \eqref{eq:nll}. With the use of the graphical LASSO algorithm from \cite{Fea2008}, the first update can be computed efficiently. The second update can be computed using non-linear optimization techniques on $\boldsymbol{\tau}$ by means of a closed-form expression for $\boldsymbol{\mu}$. In our implementation, this is performed using the L-BFGS algorithm \citep{LN1989}, which offers a super-linear convergence rate without the cumbersome evaluation and manipulation of the Hessian matrix \citep{NW2006}. The following theorem guarantees that the proposed algorithm converges to a stationary point of \eqref{eq:nll} (see Appendix B for proof).

\begin{theorem}
The BCD scheme in Algorithm \ref{alg:mle} converges to some solution $(\hat{\boldsymbol{\mu}},\hat{\bm{T}},\hat{\boldsymbol{\tau}})$ which is stationary for the penalized log-likelihood $l_\lambda(\boldsymbol{\mu}, \bm{T}, \boldsymbol{\tau})$.
\label{thm:conv}
\end{theorem}

It is worth noting that the proposed algorithm does not provide global optimization. This is not surprising, because the log-likehood $l_\lambda$ is non-convex in $\boldsymbol{\tau}$. To this end, we run multiple threads of Algorithm \ref{alg:mle} in parallel, each with a different initial point $\boldsymbol{\tau}_0$ from a large space-filling design on $[10^{-3},1-10^{-3}]^p$, then choose the converged parameter setting which yields the largest likelihood value from \eqref{eq:nll}. In our experience, this heuristic performs quite well in practice.

\section{Emulation results}\label{sec:result}

In this section, we present in four parts the emulation performance of the proposed model, when trained using the database of $n=30$ flow simulations described in Section 2. First, we briefly introduce key flow characteristics for a swirl injector, and physically interpret the flow structures extracted from CPOD. Second, we compare the numerical accuracy of our flow prediction with a validation simulation at a new injector geometry. Third, we provide a spatio-temporal quantification of uncertainty for our prediction, and discuss its physical interpretability. Lastly, we summarize the extracted flow couplings from $\mathbf{T}$, and explain why these are both intuitive and intriguing from a flow physics perspective.

\subsection{Visualization and CPOD modes}

We employ three flow snapshots of circumferential velocity (shown in Figure \ref{fig:inst}) to introduce key flow characteristics for a swirl injector: the fluid transition region, spreading angle, surface wave propagation and center recirculation. These characteristics will be used for assessing emulator accuracy, UQ and extracted flow physics. 

\begin{figure}[!t]
\begin{minipage}{0.48\textwidth}
\centering
\includegraphics[width=\linewidth]{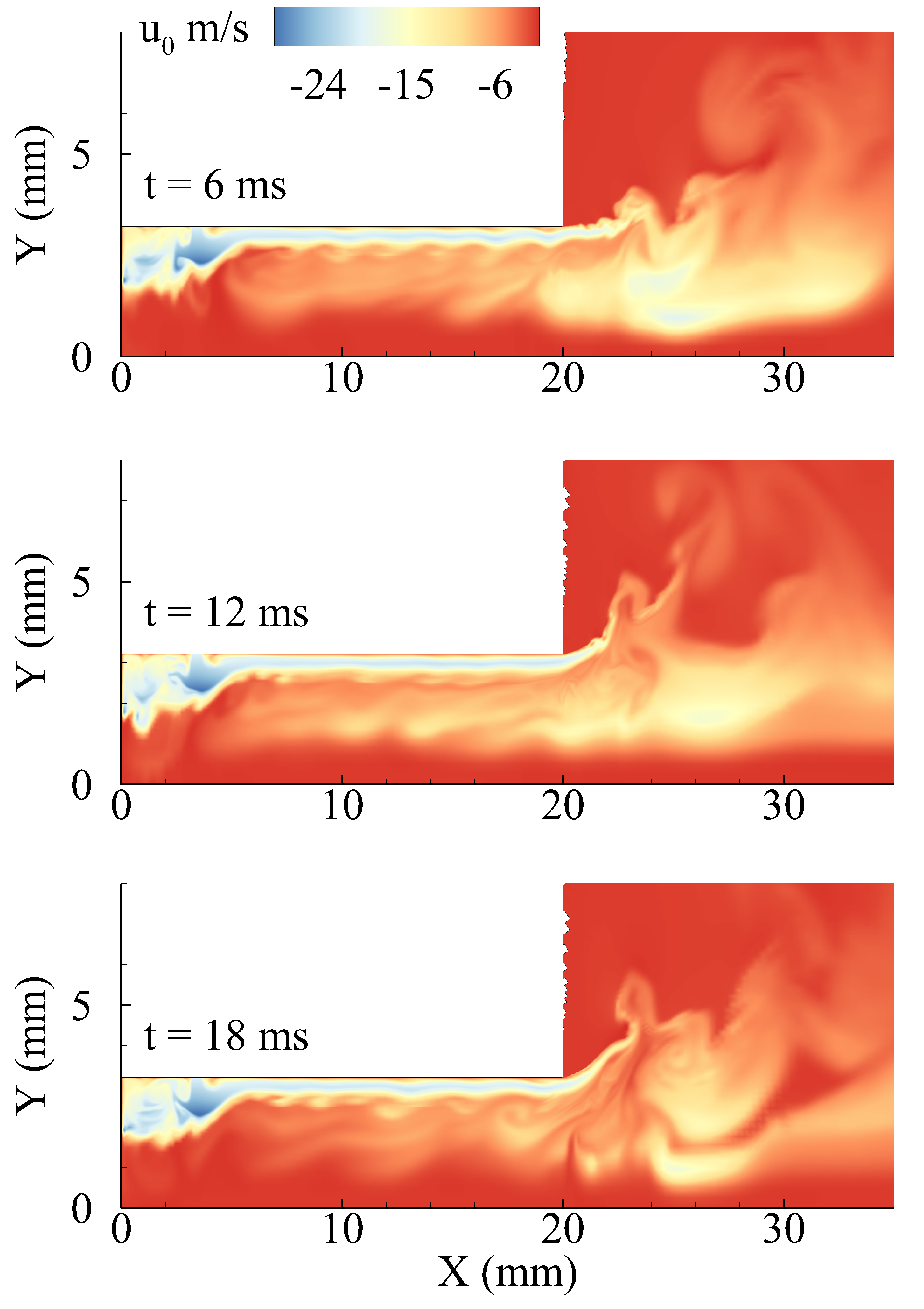}
{\caption{Flow snapshots of circumferential velocity at $t$ = 6, 12 and 18 ms.}
\label{fig:inst}}
\end{minipage}
\hfill
\begin{minipage}{0.48\textwidth}
\centering
\includegraphics[width=\textwidth]{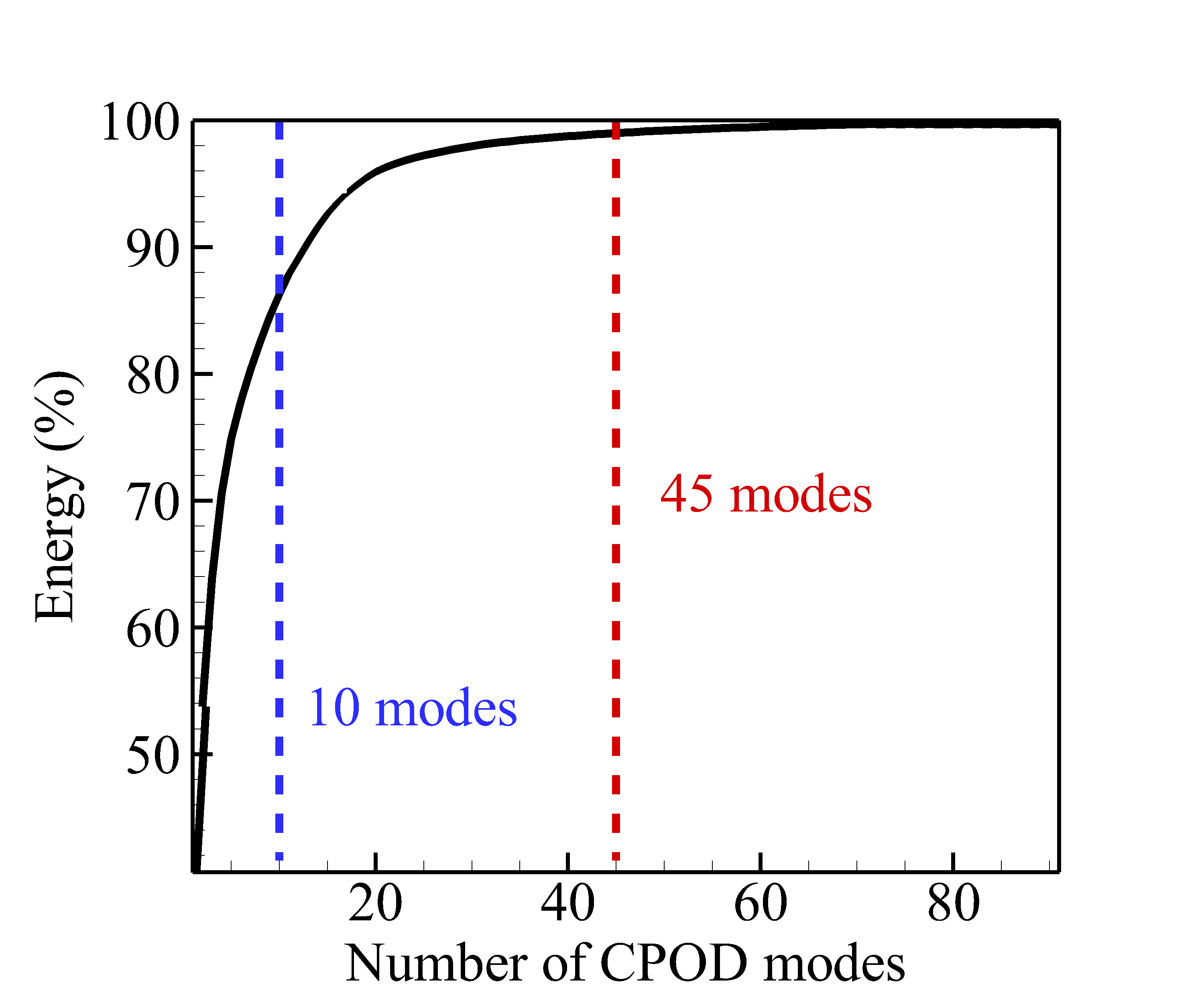}
{\caption{Energy distribution of CPOD modes for circumferential velocity flow.}
\label{fig:energy}}
\includegraphics[width=\linewidth]{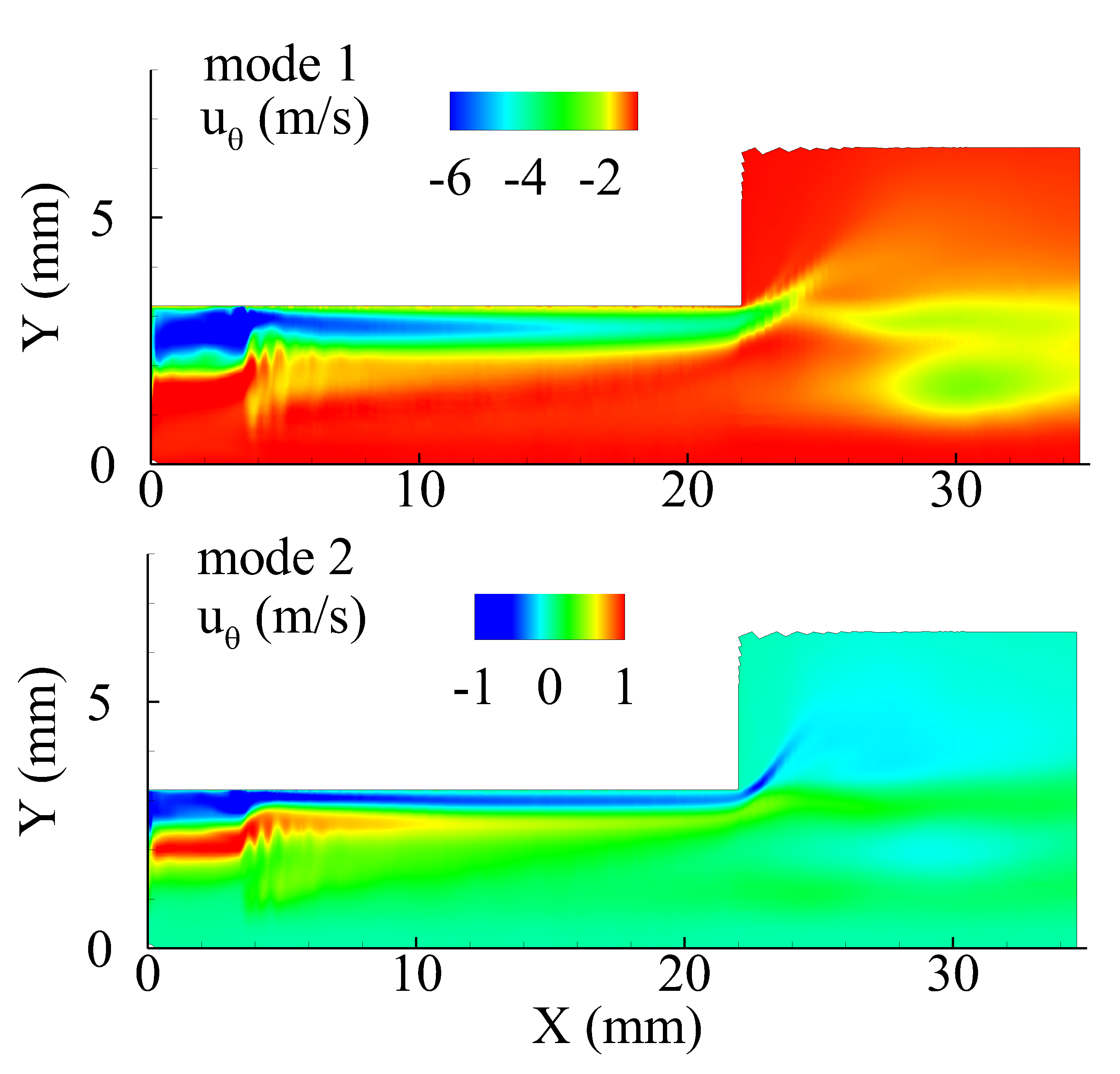}
\caption{The leading two spatial CPOD modes for circumferential velocity flow.}
\label{fig:PODmodes}
\end{minipage}
\end{figure}

\bi
\item \textit{Fluid transition region:} The fluid transition region is the region which connects compressed-liquid near the wall (colored blue in Figure \ref{fig:inst}) to light-gas (colored red) near the centerline at supercritical pressure \citep{WY2016}. This region is crucial for analyzing injector flow characteristics, as it provides the instability propagation and feedback mechanisms between the injector inlet and exit. An important emulation goal is to accurately predict both the {spatial location} of this region and its dynamics, because such information can be used to assess feedback behavior at new geometries.
\item \textit{Spreading angle:} The spreading angle $\alpha$ (along with the LOX film thickness $h$) is an important physical metric for measuring the performance of a swirl injector. A larger $\alpha$ and smaller $h$ indicate better performance of injector atomization and breakup processes. The spreading angle can be seen in Figure \ref{fig:inst} from the blue LOX flow at injector exit (see Figure \ref{fig:inj} for details).
\item \textit{Surface wave propagation:} Surface waves, which transfer energy through the fluid medium, manifest themselves as wavy structures in the flowfield. These waves allow for propagation of flow instabilities between upstream and downstream regions of the injector, and can be seen in the first snapshot of Figure \ref{fig:inst} along the LOX film boundary.
\item \textit{Center recirculation:} Center recirculation, another key instability structure, is the circular flow of a fluid around a rotational axis (this circular region is known as the {vortex core}). From the third snapshot in Figure \ref{fig:inst}, a large vortex core (in white) can be seen at the injector exit, which is expected because of sudden expansion of the LOX stream and subsequent generation of adverse pressure gradient.
\ei

Regarding the CPOD expansion, Figure \ref{fig:energy} shows the energy ratio captured using the leading $M$ terms in \eqref{eq:cklexp} for circumferential velocity, with this ratio defined as:
\[\xi(M) = \frac{\sum_{k=1}^M \sum_{i=1}^n \int \left[ \int \beta_k(t; \bm{c}_i) \mathcal{M}_i\{\phi_k(\bm{x})\} \; d\bm{x}\right]^2 \; dt}{\sum_{k=1}^\infty \sum_{i=1}^n \int \left[ \int \beta_k(t; \bm{c}_i) \mathcal{M}_i\{\phi_k(\bm{x})\} \; d\bm{x} \right]^2 \; dt}.\]
Only $M=10$ and $M=45$ modes are needed to capture 90\% and 99\% of the total flow energy over \textit{all} $n=30$ simulation cases, respectively. Compared to a similar experiment in \cite{ZY2008}, which required around $M=20$ modes to capture 99\% flow energy for a \textit{single} geometry, the current results are very promising, and show that the CPOD gives a reasonably compact representation. This also gives empirical evidence for the linearity assumption used for computation efficiency. Similar results also hold for other flow variables as well, and are not reported for brevity. Additionally, the empirical study in \cite{ZY2008} showed that the POD modes capturing the top 95\% energy have direct physical interpretability in terms of known flow instabilities. To account for these (and perhaps other) instability structures in the model, we set the truncation limit $K_r$ as the smallest value of $M$ satisfying $\xi(M) \geq 99\%$, which appears to provide a good balance between predictive accuracy and computational efficiency.

The extracted CPOD terms can also be interpreted in terms of flow physics. We illustrate this using the leading two CPOD terms for circumferential velocity, whose spatial distributions are shown in Figure \ref{fig:PODmodes}. Upon an inspection of these spatial plots and their corresponding spectral frequencies, both modes can be identified as hydrodynamic instabilities in the form of longitudinal waves propagating along the LOX film boundary. Specifically, the first mode corresponds to the first harmonic mode for this wave, and the second mode represents the second harmonic and shows the existence of an antinode in wave propagation. As we show in Section 4.4, the interpretability of CPOD modes allows the proposed model to extract physically meaningful couplings for further analysis.

\begin{figure}[!tp]
\begin{minipage}{0.48\textwidth}
\centering
\includegraphics[width=\linewidth]{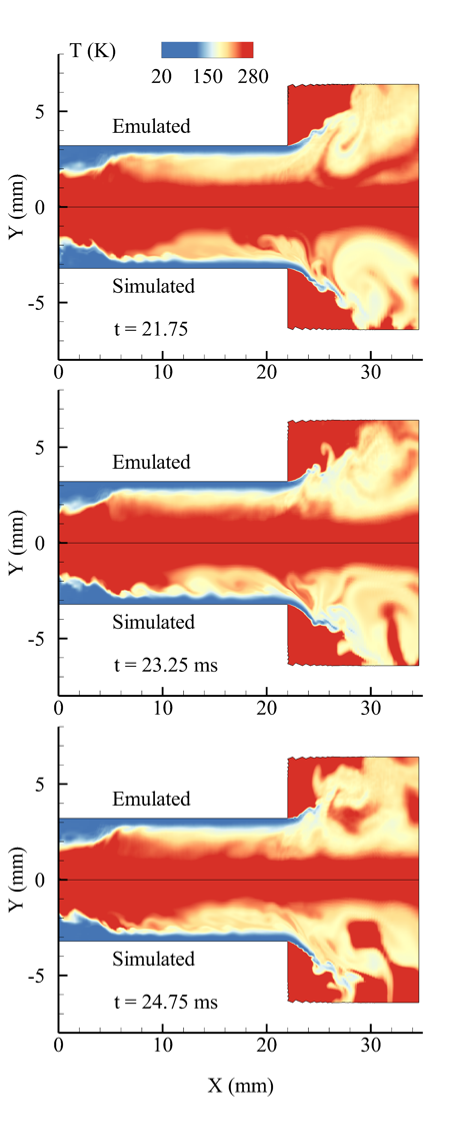}
\vspace{-1cm}
\caption{Simulated and emulated temperature flow at $t=21.75$ ms, $23.25$ ms and $24.75$ ms.}
\label{fig:comp}
\end{minipage}
\hfill
\begin{minipage}{0.48\textwidth}
\centering
\includegraphics[width=\textwidth]{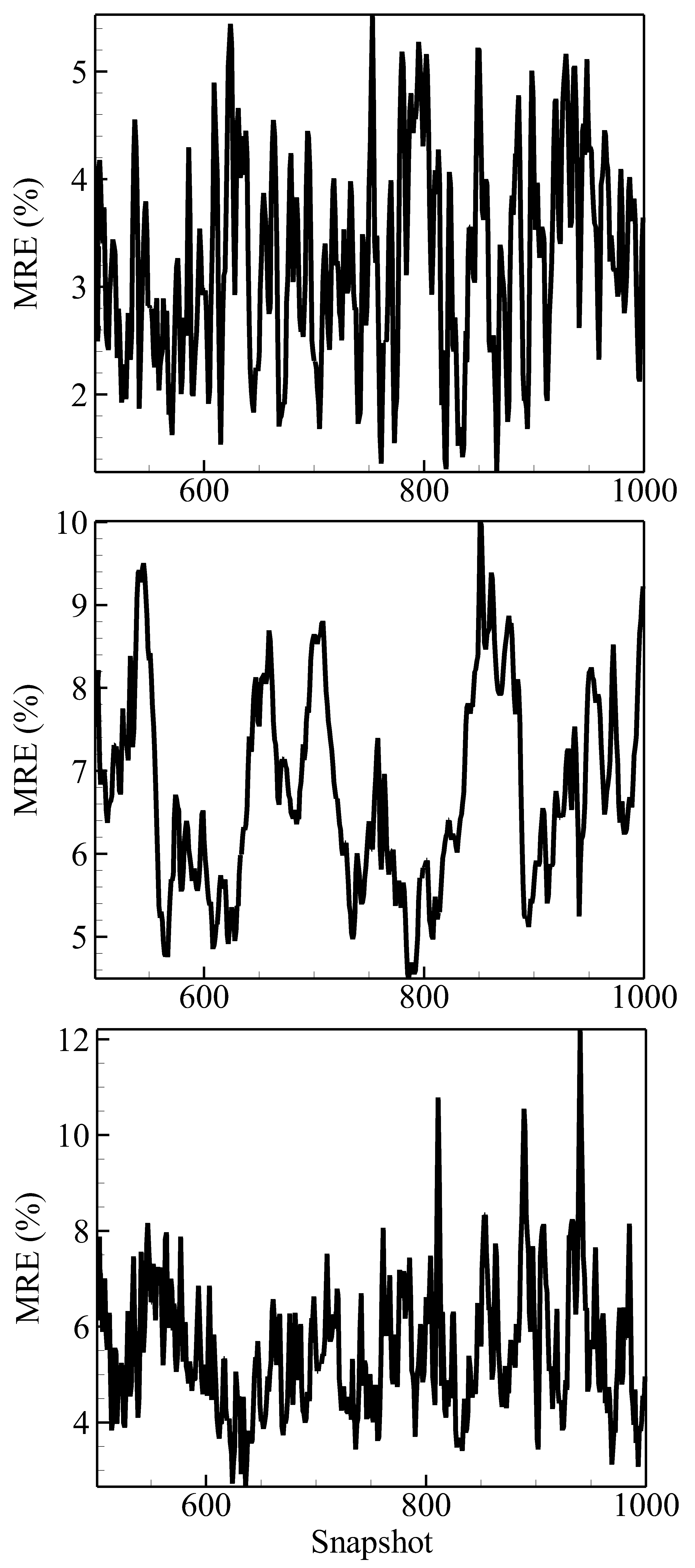}
\caption{MRE at injector inlet (top), fluid transition region (middle) and injector exit (bottom).}
\label{fig:mae}
\end{minipage}
\end{figure}

\begin{figure}[t]
\centering
\includegraphics[width=0.5\textwidth]{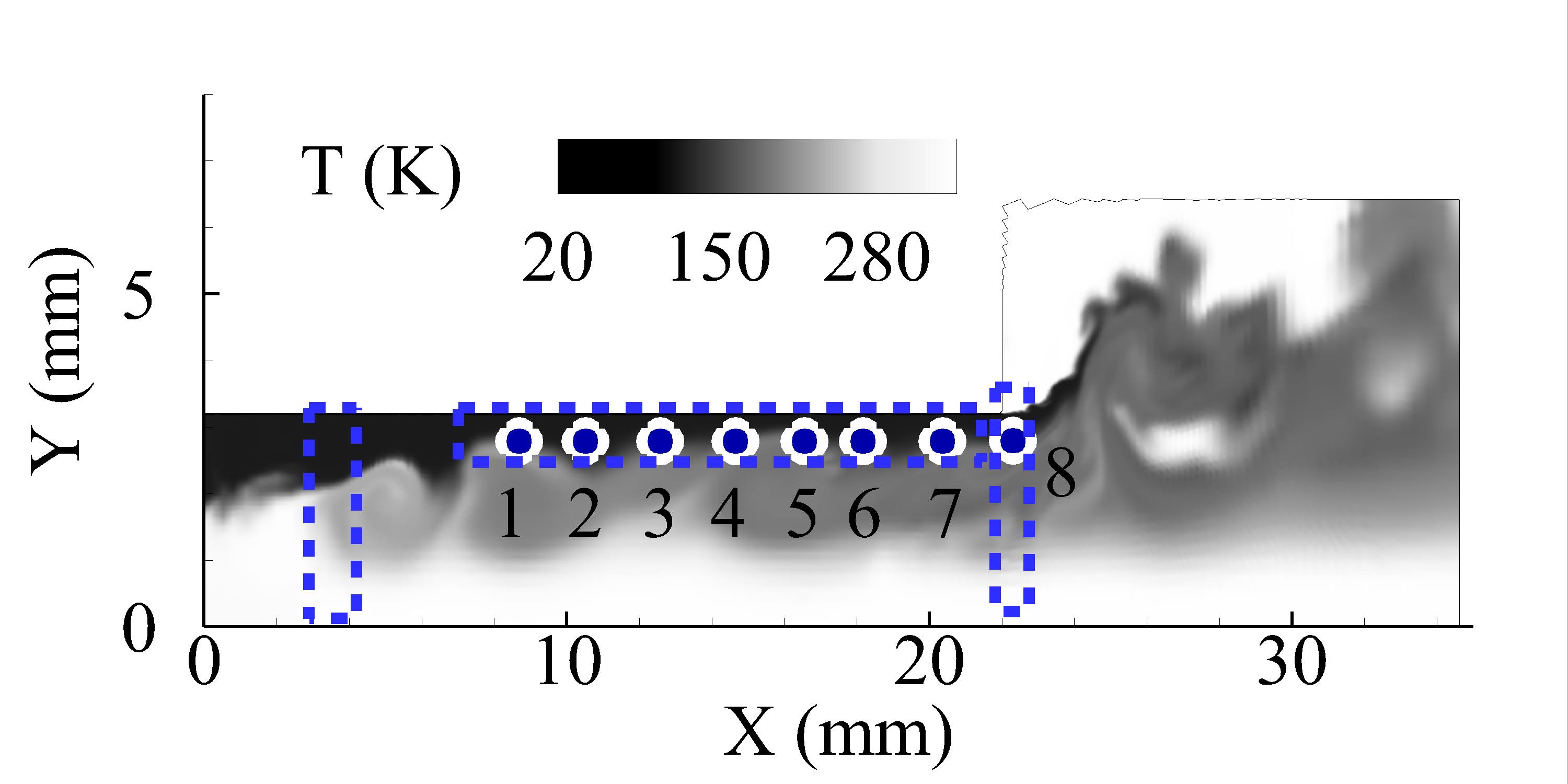}
\caption{Injector subregions (dotted in blue) and probe locations (circled in white).}
\label{fig:probes}
\end{figure}

\subsection{Emulation accuracy}
To ensure that our emulator model provides accurate flow predictions, we perform a validation simulation at the new geometric setting: $L=22\text{ mm}$, $R_n=3.215\text{ mm}$, $\Delta L=3.417\text{ mm}$, $\theta = 58.217^\circ$ and $\delta=0.576\text{ mm}$. This new geometry provides a 10\% variation on an existing injector used in the RD-0110 liquid-fuel engine \citep{Yan1995}. Since the goal is {predictive accuracy}, the sparsity penalty $\lambda$ in \eqref{eq:nll} is tuned using 5-fold cross-validation \citep{Fea2001}. We provide below a qualitative comparison of the predicted and simulated flows, and then discuss several metrics for quantifying emulation accuracy. 

Figure \ref{fig:comp} shows three snapshots of the simulated and predicted fully-developed flows for temperature, in intervals of $1.5$ ms starting at $21.75$ ms. From visual inspection, the predicted flow closely mimics the simulated flow on several performance metrics, including the fluid transition region, film thickness and spreading angle. The propagation of surface waves is also captured quite well within the injector, with key downstream recirculation zones correctly identified in the prediction as well. This comparison illustrates the effectiveness of the proposed emulator in capturing key flow physics, and demonstrates the importance of incorporating known flow properties of the fluid as assumptions in the statistical model.

Next, three metrics are used to quantify emulation accuracy. The first metric, which reports the mean relative error in important sub-regions of the injector, measures the \textit{spatial} aspect of prediction accuracy. The second metric, which inspects spectral similarities between the simulated and predicted flows, measures \textit{temporal} accuracy. The last metric investigates how well the predicted flow captures the underlying flow physics of an injector.

For {spatial} accuracy, the following mean relative error (MRE) metric is used:
\[
\text{MRE}(t;\mathcal{S})=\frac{\int_{\mathcal{S}}|Y(\mathbf{x},t;\mathbf{c}_{new})-\hat{Y}(\mathbf{x},t;\mathbf{c}_{new})| \; d \bm{x}}{\int_{\mathcal{S}}|Y(\mathbf{x},t;\mathbf{c}_{new})| \; d\bm{x}}\times 100\%,
\]
where $Y(\mathbf{x},t;\mathbf{c}_{new})$ is the simulated flow at setting $\bm{c}_{new}$, and $\hat{Y}(\mathbf{x},t;\mathbf{c}_{new})$ is the flow predictor in \eqref{eq:flowpred} (for brevity, the superscript for flow variable $r$ is omitted here). In words, MRE($t;\mathcal{S}$) provides a measure of emulation accuracy within a desired sub-region $\mathcal{S}$ at time $t$, relative to the overall flow energy in $\mathcal{S}$. Since flow behaviors within the injector inlet, fluid transition region and injector exit (outlined in Figure \ref{fig:probes}) are crucial for characterizing injector instability, we investigate the MRE specifically for these three sub-regions. Figure \ref{fig:mae} plots MRE($t,\mathcal{S}$) for $t = 15 - 30$ ms, when the flow has fully developed. For all three sub-regions, the relative error is within a tolerance level of 10\% for nearly all time-steps, which is very good from an engineering perspective. 

\begin{figure}[!t]
\centering
\includegraphics[width=\linewidth]{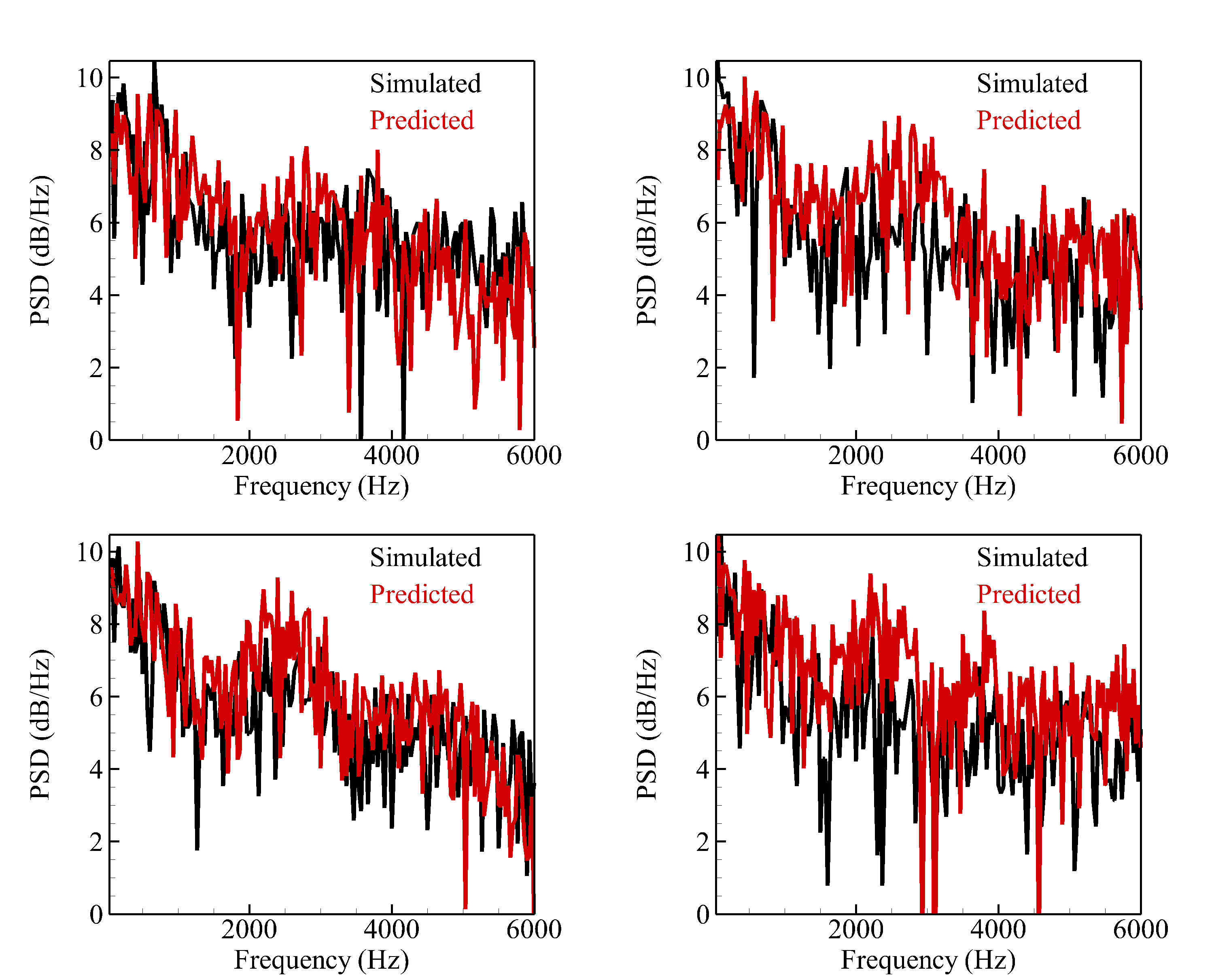}
\caption{PSD spectra for pressure flow at probes 1, 3, 5 and 7 (see Figure \ref{fig:probes}).}
\label{fig:PSD}
\end{figure}

To assess {temporal} accuracy, we conduct a power spectral density (PSD) analysis of predicted and simulated pressure flows at eight specific probes along the region of surface wave propagation (see Figure \ref{fig:probes}). This analysis is often performed as an empirical tool for assessing injector stability (see \citealp{ZY2008}), because surface waves allow for feedback loops between upstream and downstream oscillations \citep{bazarov1998liquid}. Figure \ref{fig:PSD} shows the PSD spectra for the predicted and simulated flow at four of these probes. Visually, the spectra look very similar, both at low and high frequencies, with {peaks} nearly identical for the predicted and simulated flow. Such peaks are highly useful for analyzing flow physics, because they can be used to {identify} physical properties (e.g., hydrodynamic, acoustic, etc.) of dominant instability structures. In this sense, the proposed emulator does an excellent job in mimicking important {physics} of the simulated flow.

\begin{figure}[!t]
\begin{minipage}{0.48\textwidth}
\centering
\includegraphics[width=\textwidth]{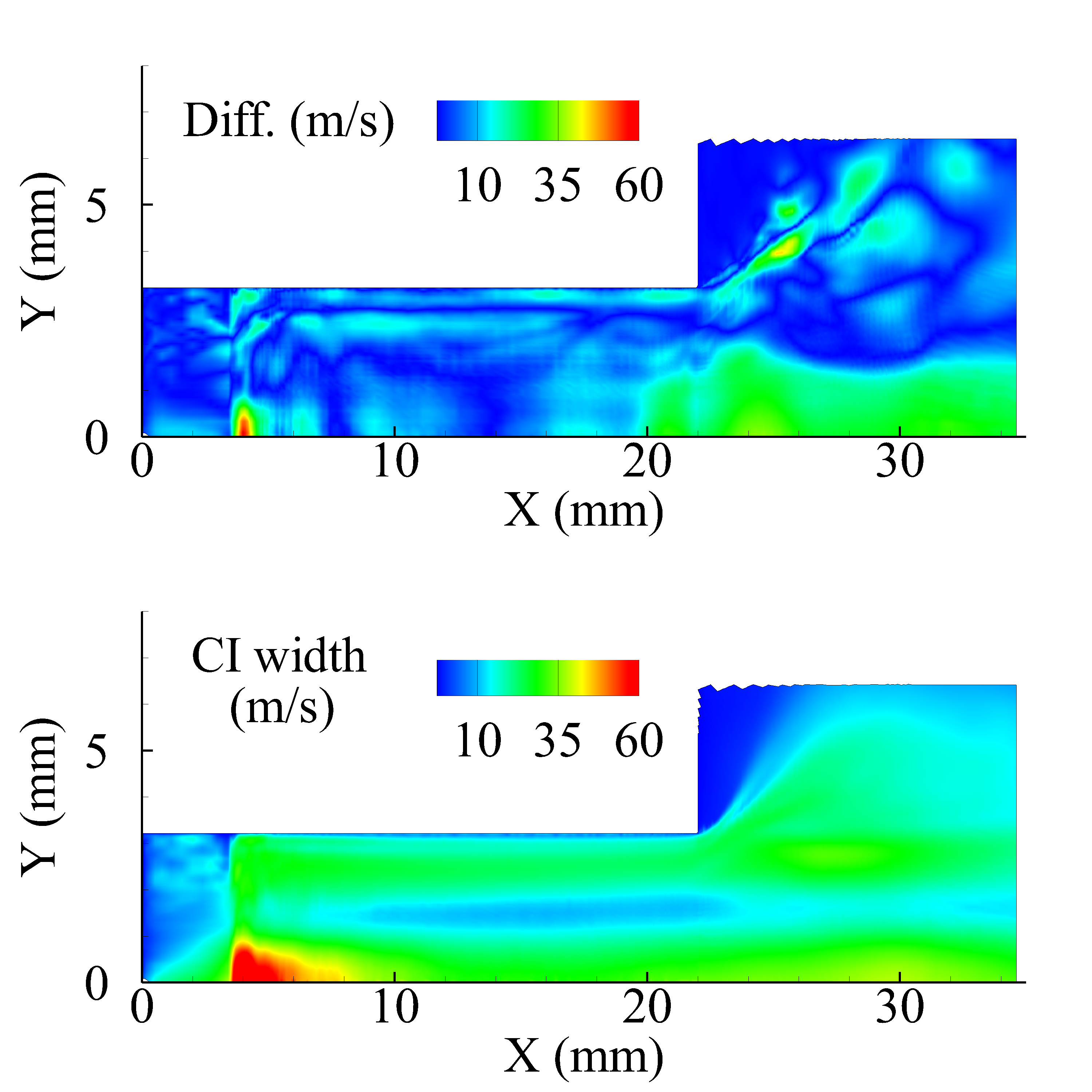}
\caption{Absolute prediction error (top) and pointwise CI width (bottom) for $x$-velocity at $t=15$ ms.}
\label{fig:UQcomparison}
\end{minipage}
\hfill
\begin{minipage}{0.48\textwidth}
\centering
\includegraphics[width=\textwidth]{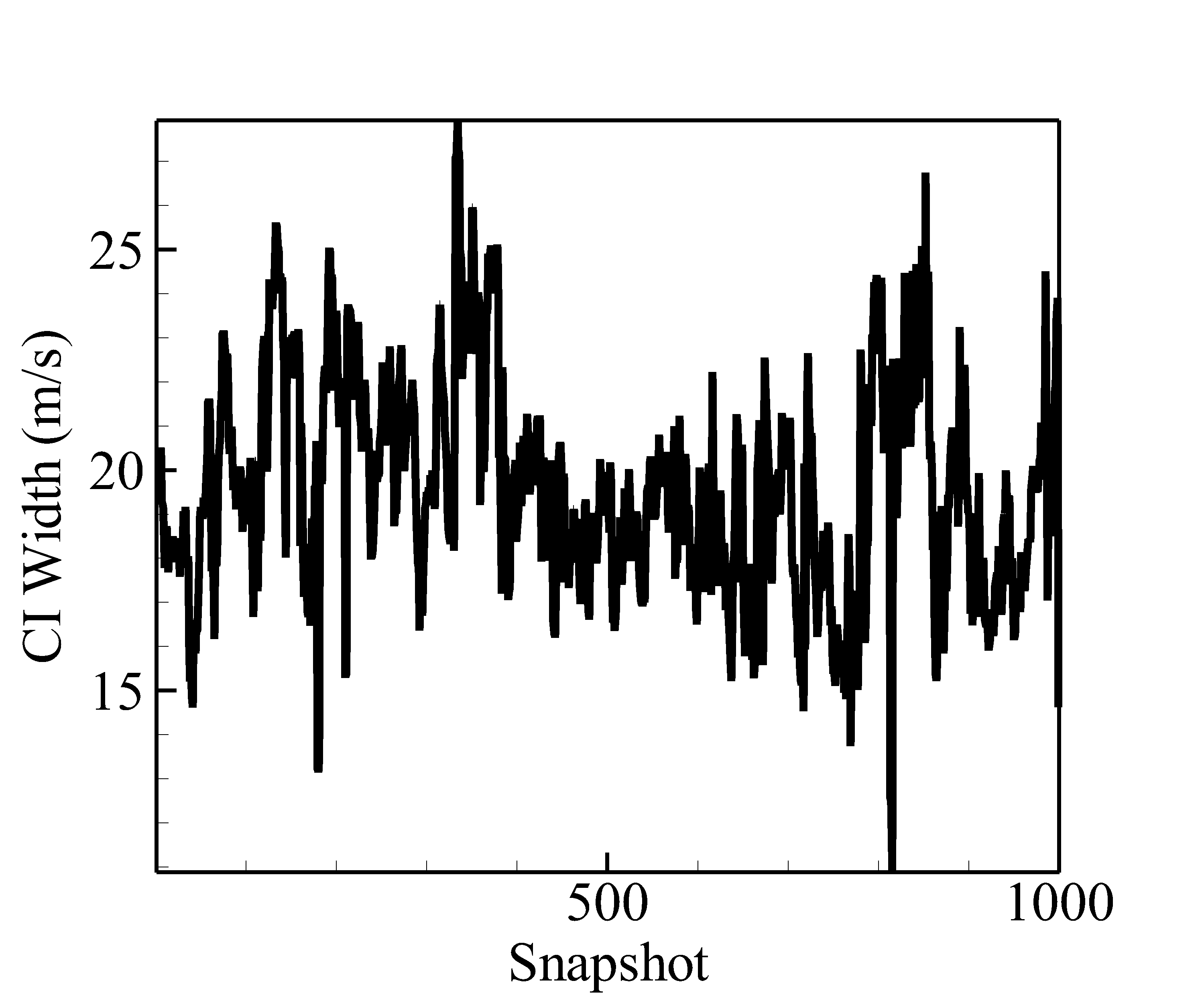}
\caption{CI width of $x$-velocity at probe 1.}
\label{fig:UQtemporal}
\end{minipage}
\end{figure}

Finally, we investigate the film thickness $h$ and spreading angle $\alpha$, which are key performance metrics for injector performance. Since both of these metrics are computed using spatial gradients of flow variables, an accurate emulation of these measures suggests accurate flow emulation as well. For the validation setting, the simulated (predicted) flow has a film thickness of 0.47 mm (0.42 mm) and a spreading angle of 103.63$^\circ$ (107.36$^\circ$), averaged over the fully-developed timeframe from $t=15 - 30$ ms. This corresponds to relative errors of 10.6\% and 3.60\%, respectively, and is within the desired error tolerance from an engineering perspective.

\subsection{Uncertainty quantification}

For computer experiments, the quantification of predictive uncertainty can be as important as the prediction itself. To this end, we provide a spatio-temporal representation of this UQ, and show that it has a useful and appealing physical interpretation. For spatial UQ, the top plot of Figure \ref{fig:UQcomparison} shows the {one-sided width} of the 80\% pointwise confidence interval (CI) from \eqref{eq:flowvar} for $x$-velocity at $t=15$ ms. It can be seen that the emulator is most certain in predicting near the inlet and centerline of the injector, but shows high predictive uncertainty at the three gaseous cores downstream (in green). This makes physical sense, because these cores correspond to flow recirculation vortices, and therefore exhibit highly unstable flow behavior. From the bottom plot of Figure \ref{fig:UQcomparison}, which shows the absolute emulation error of the same flow, the pointwise confidence band not only covers the realized prediction error, but roughly mimics its spatial distribution as well.

For temporal UQ, Figure \ref{fig:UQtemporal} shows the same one-sided CI width at probe 1 (see Figure \ref{fig:probes}). We see that this temporal uncertainty is relatively steady over $t$, except for two abrupt spikes at time-steps around 300 and 800. These two spikes have an appealing physical interpretation: the first indicates a {flow displayment} effect of the central vortex core, whereas the second can be attributed to the {boundary development} of the same core. This again demonstrates the usefulness of UQ not only as a measure of predictive uncertainty, but also as a means for extracting useful flow physics without the need for expensive simulations.

To illustrate the improved UQ of the proposed model (see Theorem \ref{thm:uq}), we use a derived quantity called {turbulent kinetic energy} (TKE). TKE is typically defined as:
\small
\begin{equation}
\kappa(\bm{x},t) = \frac{1}{2}\sum_{r \in \{u,v,w\}}\left\{{Y}^{(r)}(\mathbf{x},t)-\bar{{Y}}^{(r)}(\mathbf{x})\right\}^2,
\label{eq:kedef}
\end{equation}
\normalsize
where $Y^{(u)}(\mathbf{x},t),Y^{(v)}(\mathbf{x},t)$ and $Y^{(w)}(\mathbf{x},t)$ are flows for $x$-, $y$- and circumferential velocities, respectively, with $\bar{Y}^{(u)}(\mathbf{x}),\bar{Y}^{(v)}(\mathbf{x})$ and $\bar{Y}^{(w)}(\mathbf{x})$ its corresponding time-averages. Such a quantity is particularly important for studying turbulent instabilities, because it measures {fluid rotation energy} within eddies and vortices.

For the sake of simplicity, assume that (a) the time-averages $\bar{Y}^{(u)}(\mathbf{x}),\bar{Y}^{(v)}(\mathbf{x})$ and $\bar{Y}^{(w)}(\mathbf{x})$ are fixed, and (b) the parameters $(\boldsymbol{\mu}, \bm{T}, \boldsymbol{\tau})$ are known. The following theorem provides the MMSE predictor and pointwise confidence interval for $\kappa(\bm{x},t)$ (proof in Appendix C).
\begin{theorem}
For fixed $\bm{x}$ and $t$, the MMSE predictor of $\kappa(\bm{x},t)$ at a new setting $\bm{c}_{new}$ is
\small
\begin{equation}
\hat{\kappa}(\bm{x},t) = \frac{1}{2}\sum_{r \in \{u,v,w\}}\left\{\hat{Y}^{(r)}(\mathbf{x},t)-\bar{{Y}}^{(r)}(\mathbf{x})\right\}^2 + \text{tr}\{\Phi(\bm{x},t)\},
\label{eq:kepred}
\end{equation}
\normalsize
where $\hat{Y}^{(u)}(\mathbf{x},t),\hat{Y}^{(v)}(\mathbf{x},t)$ and $\hat{Y}^{(w)}(\mathbf{x},t)$ are {predicted} flows for $x$-, $y$- and circumferential velocities from \eqref{eq:flowpred}, and $\Phi(\bm{x},t)$ is defined in (C.1) of Appendix C. Moreover, $\hat{\kappa}(\bm{x},t)$ is distributed as a weighted sum of non-central $\chi^2$ random variables, with an explicit expression given in (C.3) of Appendix C.
\label{thm:kinuq}
\end{theorem}
\noindent In practice, plug-in estimates are used for both time-averaged flows and model parameters.

\begin{figure}[t]
\centering
\includegraphics[width=\textwidth]{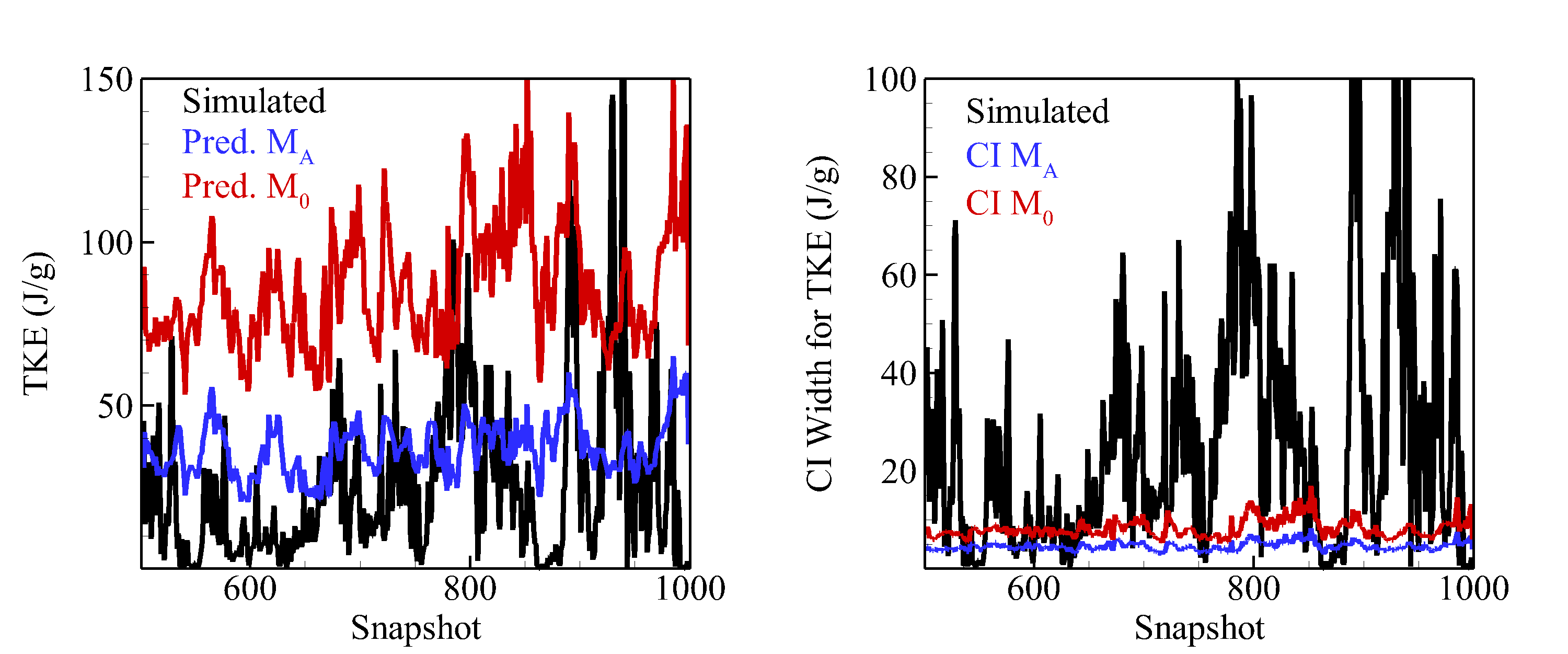}
\caption{Predicted TKE and lower 90\% confidence band for $M_A$ and $M_0$ at probe 8.}
\label{fig:ke}
\end{figure}

With this in hand, we compare the prediction and UQ of TKE from the proposed model $M_A$ and the independent model $M_0$ (see Theorem \ref{thm:uq}) with the simulated TKE at the validation setting. Figure \ref{fig:ke} shows the predicted TKE $\hat{\kappa}(\bm{x},t)$ at probe 8 over the fully-developed time-frame of $t= 15 - 30$ ms, along with the 90\% lower pointwise confidence band constructed using Theorem \ref{thm:kinuq}. Visually, the proposed model $M_A$ provides an improved prediction of the simulated TKE than the independent model $M_0$. As for the confidence bands, the average coverage rate for $M_A$ over the fully-developed time-frame (85.0\%) is much closer to the desired nominal rate of 90\% compared to that for $M_0$ (73.8\%). The proposed model therefore provides a coverage rate closer to the desired nominal rate of 90\%. The poor coverage rate for the independent model is shown in the right plot of Figure \ref{fig:ke}, where the simulated TKE often dips below the lower confidence band. By incorporating prior knowledge of flow couplings, the proposed model can provide improved predictive performance and uncertainty quantification.

\begin{table}[t]
\begin{minipage}{0.48\textwidth}
\centering
\includegraphics[width=\textwidth]{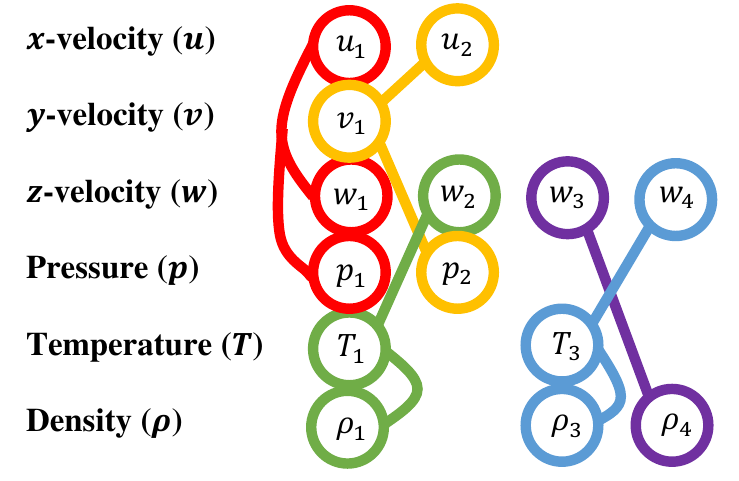}
\captionof{figure}{Graph of selected flow couplings from $\bm{T}$. Nodes represent CPOD modes, and edges represent non-zero correlations.}
\label{fig:T}
\end{minipage}
\hfill
\begin{minipage}{0.48\textwidth}
\begin{tabular}{cc}
\toprule
\text{\bf{Step}} & \text{\bf{Comp. time (mins)}}\\
\toprule
\text{CPOD extraction} & 33.91\\
\text{Parameter estimation} & 11.31\\
\text{Flow prediction} & 20.19\\
\toprule
\text{Total} & 65.41\\
\toprule
\end{tabular}
\caption{Computation time for each step of the proposed emulator, parallelized over 200 processing cores.}
\label{tbl:comptime}
\end{minipage}
\end{table}

\subsection{Correlation extraction}
\label{sec:correx}
Finally, we demonstrate the use of the proposed model as a tool for extracting {common flow couplings} on the design space. Setting the sparsity penalty $\lambda$ so that only the top nine correlations are chosen, Figure \ref{fig:T} shows the corresponding graph of the extracted couplings of CPOD modes. Nodes on this graph represent CPOD modes for each flow variable, with edges indicating the presence of a non-zero correlation between two modes. Each connected subgraph in Figure \ref{fig:T} is interpretable in terms of flow physics. For example, the subgraph connecting $u_1$, $w_1$ and $P_1$ (first modes for $x$-velocity, circumferential velocity and pressure) makes physical sense, because $u_1$ and $w_1$ are inherently coupled by Bernoulli's equation for fluid flow \citep{SS1982}, while $w_1$ and $P_1$ are connected by the centrifugal acceleration induced by circular momentum of LOX flow. Likewise, the subgraph connecting $T_1$, $\rho_1$ and $w_2$ also provides physical insight: $T_1$ and $\rho_1$ are coupled by the equation of state and conservation of energy, while $\rho_1$ and $w_2$ are connected by conservation of momentum.

The interpretability of these extracted flow couplings in terms of fundamental conservation laws from fluid mechanics is not only appealing from a flow physics perspective, but also provides a reassuring check on the estimation of the co-kriging matrix $\bm{T}$. Recall from the discussion in Section \ref{sec:covmat} that an accurate estimate of $\bm{T}$ is needed for the improved predictive guarantees of Theorem \ref{thm:uq} to hold. The consistency of the selected flow couplings (and the ranking of such couplings) with established physical principles provides confidence that the proposed estimation algorithm indeed returns an accurate estimate of $\bm{T}$. These results nicely illustrate the dual purpose of the CPOD matrix $\bm{T}$ in our co-kriging model: not only does it allow for more accurate UQ, it also extracts interesting flow couplings which can guide further experiments.

\subsection{Computation time}

In addition to accurate flow emulation and physics extraction, the primary appeal of the proposed emulator is its efficiency. Table \ref{tbl:comptime} summarizes the computation time required for each step of the emulation process, with timing performed on a parallelized system of 200 Intel Xeon E5-2603 1.80GHz processing cores. Despite the massive training dataset, which requires nearly 100GB of storage space, we see that the proposed model can provide accurate prediction, UQ and coupling extraction in slightly over an hour of computation time. Moreover, because both CPOD extraction and parameter estimation need to be performed only once, the surrogate model can generate flow predictions for hundreds of new settings within a day's time, thereby allowing for the exploration of the full design space in practical turn-around times. Through a careful elicitation and incorporation of flow physics into the surrogate model, we show that an efficient and accurate flow prediction is possible despite a limited number of simulation runs, with the trained model extracting valuable physical insights which can be used to guide further investigations.


\section{Conclusions and future work}\label{sec:concl}

In this paper, a new emulator model is proposed which efficiently predicts turbulent cold-flows for rocket injectors with varying geometries. An important innovation of our work lies in its \textit{elicitation} and \textit{incorporation} of flow properties as model assumptions. First, exploiting the deep connection between POD and turbulent flows \citep{Lum1967}, a novel CPOD decomposition is used for extracting {common} instabilities over the design space. Next, taking advantage of dense temporal resolutions, a {time-independent} emulator is proposed that considers {independent} emulators at each simulation time-step. Lastly, a sparse covariance matrix $\bm{T}$ is employed within the emulator model to account for the few significant couplings among flow variables. Given the complexities inherent in spatio-temporal flows and the massive datasets at hand, such simplifications are paramount for {accurate} flow predictions in {practical} turn-around times. This highlights the need for careful {elicitation} in flow emulation, particularly for engineering applications where the time-consuming nature of simulations limits the number of available runs.

Applying the model to simulation data, the proposed emulator provides accurate flow predictions and captures several key metrics for injector performance. In addition, the proposed model offers two appealing features: (a) it provides a physically meaningful quantification of spatio-temporal uncertainty, and (b) it extracts significant couplings between flow instabilities. A key advantage of our emulator over existing flow kriging methods is that it provides accurate predictions using only  a fraction of the time required by simulation. This efficiency is very appealing for engineers, because it allows them to fully explore the desired design space and make timely decisions.

Looking ahead, we are pursuing several directions for future research. First, while the CPOD expansion appears to work well for cold-flows, the justifying assumption of similar Reynolds numbers does not hold for more complicated (e.g., reacting) turbulent flows. To this end, we are working on ways to incorporate pattern recognition techniques \citep{Fuk2013} into the GP kriging framework to jointly (a) identify common instability structures that scale {non-linearly} over varying geometries, then (b) predict such structures at new geometric settings. The key hurdle is again {computational efficiency}, and the treed GP models in \cite{Tea2012} or the local GP models in \cite{GA2015} and \cite{Sea2016} appear to be attractive options. Next, a new design is proposed recently in \cite{MJ2017} which combines the MaxPro methodology with minimax coverage, and it will be interesting to see whether such designs can provide improved performance. Lastly, to evaluate the stability of new injector geometries, the UQ for the emulated flow needs to be fed forward through an acoustics solver. Since each evaluation of the solver can be time-intensive, this forward uncertainty propagation can be performed more quickly by reducing this UQ to a set of representative points, and the support points in \cite{MJ2016} can prove to be useful for conducting such a task.\\

\if1\blind{\noindent \textbf{Acknowledgements}: The authors gratefully acknowledge helpful advice from the associate editor, two anonymous referees and Dr. Mitat A. Birkan. This work was sponsored partly by the Air Force Office of Scientific Research under Grant No. FA 9550-10-1-0179, and partly by the William R. T. Oakes Endowment of Georgia Institute of Technology. Wu's work is partially supported by NSF DMS 1564438.}
\fi

\newpage
\begin{appendices}

\numberwithin{equation}{section}
\counterwithin{figure}{section} 
\counterwithin{table}{section} 
\setcounter{page}{1}

\section{Computing the CPOD expansion}
\label{sec:CPOD}
The driving idea behind CPOD is that a common spatial domain is needed to extract common instabilities over multiple injector geometries, since each simulation run has different geometries and varying grid points. We first describe a physically justifiable method for obtaining such a common domain, and then use this to compute the CPOD expansion.

\subsection{Common grid}\label{sec:commongrid}


\begin{figure}[t]
\centering
\includegraphics[width=\textwidth]{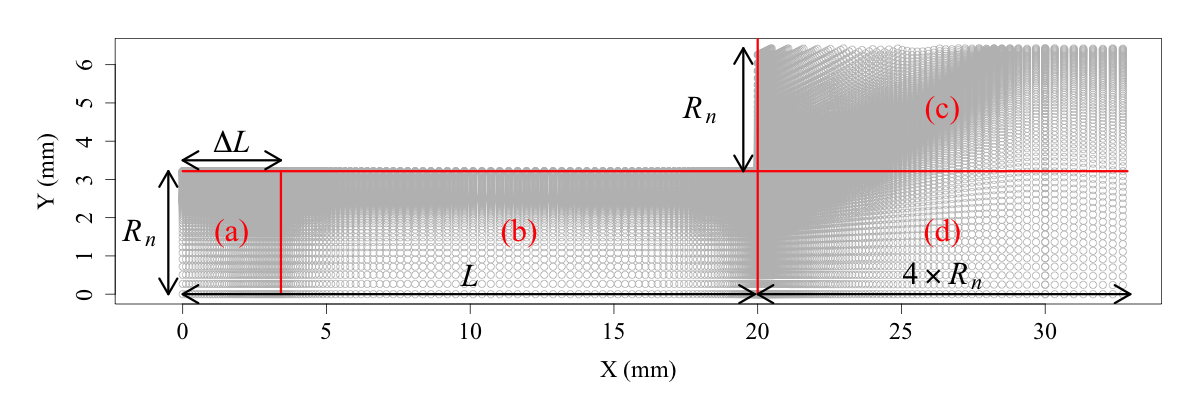}
\caption{Partition of the spatial grid for the first simulation case.}
\label{fig:FourParts}
\end{figure}

\ben
\item Identify the densest grid (i.e., with the most grid points) among the $n$ simulation runs, and set this as the common reference grid.
\item For each simulation, partition the grid into the following four parts: (a) from injector head-end to the inlet, (b) from the inlet to the nozzle exit, (c) the top portion of the downstream region and (d) the bottom portion of the downstream region (see Figure \ref{fig:FourParts} for an illustration). This splits the flow in such a way that the linearity assumption can be physically justified.
\item Linearly rescale each part of the partition to the common grid by the corresponding geometry parameters $L$, $R_n$ and $\Delta L$ (see Figure \ref{fig:FourParts}).
\item For each simulation, interpolate the original flow data onto the spatial grid of the common geometry. This step ensures the flow is realized over a common set of grid points for all $n$ simulations. In our implementation, the \textit{inverse distance weighting} interpolation method \citep{She1968} is used with $10$ nearest neighbours.
%
\een

\subsection{POD expansion}
After flows from each simulation have been rescaled onto the common grid, the original POD expansion can be used to extract common flow instabilities. Let $\{\bm{x}_j\}_{j=1}^J$ and $\{t_m\}_{m=1}^T$ denote the set of common grid points and simulated time-steps, respectively, and let $\tilde{Y}(\bm{x},t;\bm{c}_i)$ be an interpolated flow variable for geometric setting $\bm{c}_i$, $i=1, \cdots, n$ (for brevity, assume a single flow variable, e.g., $x$-velocity, for the exposition below). The CPOD expansion can be computed using the following three steps.

\ben
\item For notational convenience, we combine all combinations of geometries and time-steps into a single index. Set $N=nT$ and let $l = 1, \cdots, N$ index all combinations of $n$ design settings and $T$ time-steps, and let $\tilde{Y}_l(\bm{x}) \equiv \tilde{Y}(\bm{x},(t,\bm{c})_l)$. Define $\bm{Q} \in \mathbb{R}^{N \times N}$ as the following inner-product matrix:
\[\bm{Q}_{l,m} = \sum_{j=1}^J \tilde{Y}_l (\bm{x}_j) \tilde{Y}_m (\bm{x}_j).\]
Such an inner-product is possible because all $n$ simulated flows are observed on a set of \textit{common} gridpoints set.

First, compute the eigenvectors $\bm{a}_k \in \mathbb{R}^{N}$ satisfying:
\[ \bm{Q} \bm{a}_k = \lambda_k \bm{a}_k,\]
where $\lambda_k$ is the $k$-th largest eigenvalue of $\bm{Q}$. Since a full eigendecomposition requires $O(N^3)$ work, this step may be intractible to perform when the temporal resolution is dense. To this end, we employed a variant of the implicitly restarted Arnoldi method \citep{LS1998}, which can efficiently approximate {leading} eigenvalues and eigenvectors.

\item Compute the $k$-th mode $\phi_k(\bm{x})$ as:
\[
\begin{bmatrix}
\phi_k(\bm{x}_1) \\
\phi_k(\bm{x}_2) \\
\vdots \\
\phi_k(\bm{x}_J)
\end{bmatrix}
 = 
\begin{pmatrix}
\tilde{Y}_1 (\bm{x}_1) & \cdots & \tilde{Y}_{N} (\bm{x}_1)\\
\vdots & \ddots & \vdots \\
\tilde{Y}_1 (\bm{x}_J) & \cdots & \tilde{Y}_{N} (\bm{x}_J)
\end{pmatrix}
\bm{a}_k.
\]
To ensure orthonormality, apply the following normalization:
\[\phi_k(\bm{x}_j) :=  \frac{\phi_k(\bm{x}_j)}{\|\phi_k(\bm{x})\|}, \quad \|\phi_k(\bm{x})\| = \sqrt{\sum_{j=1}^J \phi_k(\bm{x}_j)^2}\]

\item Lastly, derive the CPOD coefficients $(\beta_{l,1}, \cdots, \beta_{l,N})^T$ for the snapshot at index $l$ (i.e., with design setting and time-step $(\bm{c},t)_l$) as:
\[
\begin{bmatrix}
\beta_{l,1} \\
\beta_{l,2} \\
\vdots \\
\beta_{l,N}
\end{bmatrix}
 = 
\begin{pmatrix}
\phi_1(\bm{x}_1) & \cdots & \phi_1(\bm{x}_J)\\
\vdots & \ddots & \vdots \\
\phi_N(\bm{x}_1) & \cdots & \phi_N(\bm{x}_J)
\end{pmatrix}
\begin{bmatrix}
\tilde{Y}_l (\bm{x}_1)\\
\tilde{Y}_l (\bm{x}_2)\\
\vdots\\
\tilde{Y}_l (\bm{x}_J)
\end{bmatrix}.
\]
Using these coefficients and a truncation at $K_r < N$ modes, it is easy to show the following decomposition of the flow at the design setting $\bm{c}_i$ and time-step $t_m$ indexed by $l$:
\[
Y(\bm{x}_j,t_m;\bm{c}_i)\approx\sum_{k=1}^{K_r} \beta_{l,k}\mathcal{M}_i\{\phi_k(\bm{x}_j)\},\quad j = 1, \cdots , J,
\]
as asserted in (3).
\een

\section{Proof of Theorem 2}
\label{sec:pf2}
Define the map $A: \mathbb{R}^K \times \mathbb{R}^{K \times K} \times \mathbb{R}^p \rightarrow \mathbb{R}^K \times \mathbb{R}^{K \times K} \times \mathbb{R}^p$ as a single-loop of the graphical LASSO operator for optimizing $\bm{T}$ with $\boldsymbol{\mu}$ and $\boldsymbol{\tau}$ fixed, and define $B:\mathbb{R}^K \times \mathbb{R}^{K \times K} \times \mathbb{R}^p \rightarrow \mathbb{R}^K \times \mathbb{R}^{K \times K} \times \mathbb{R}^p$ as the L-BFGS map for a single line-search when optimizing $\boldsymbol{\mu}$ and $\boldsymbol{\tau}$ with $\bm{T}$ fixed. Each BCD cycle in Algorithm 1 then follows the map composition $S = A^M \circ B^N$, where $M < \infty$ and $N < \infty$ are the iteration count for the graphical LASSO operator and number of line-searches, respectively. The parameter estimates at iteration $m$ of the BCD cycle can then be given by:
\[\Theta_{m+1} = S(\Theta_m), \quad \text{where} \; \Theta_m = (\boldsymbol{\mu}_m,\bm{T}_m,\boldsymbol{\tau}_m).\]

Define the set of stationary solutions as $\Gamma = \{\Theta \; : \; \nabla l_{\lambda}(\Theta) = \bm{0}\}$, where $\nabla l_\lambda$ is the gradient of the negative log-likelihood $l_\lambda$. Using the Global Convergence Theorem (see Section 7.7 of \citealp{LY2008}), we can prove stationary convergence:
\[\lim_{m\rightarrow \infty}\Theta_m = \Theta^* \in \Gamma,\]
if the following three conditions hold:
\begin{enumerate}[label=(\roman*)]
\item $\{\Theta_m\}_{m=1}^\infty$ is contained within a compact subset of $\mathbb{R}^K \times \mathbb{R}^{K \times K} \times \mathbb{R}^p$,
\item $l_\lambda$ is a continuous descent function on $\Gamma$ under map $S$,
\item $S$ is closed for points outside of $\Gamma$.
\end{enumerate}
We will verify these conditions below.
\ben[label=(\roman*)]
\item This is easily verified by the fact that $|\boldsymbol{\mu}_m| \leq \left( \max_{i,r,k} |\beta_k^{(r)}(\bm{c}_i)| \right) \bm{1}_K$, $\bm{0} \preceq \bm{T}_m \preceq \left( \max_{k,r} s^2\{\beta_k^{(r)}(\bm{c}_i)\}_{i=1}^n \right) \bm{I}_K$ and $\boldsymbol{\tau}_m \in [0,1]^p$, where $s^2\{ \cdot \}$ returns the sample standard deviation for a set of scalars.
\item To prove that $S$ is a descent function, we need to show that if $\Theta \in \Gamma$, then $l_\lambda\{S(\Theta)\} = l_\lambda\{\Theta\}$, and if $\Theta \notin \Gamma$, then $l_\lambda\{S(\Theta)\} < l_\lambda\{\Theta\}$. The first condition is trivial, since $M=0$ and $N=0$ when $\Theta$ is stationary. The second condition follows from the fact that the maps $A$ and $B$ incur a strict decrease in $l_\lambda$ whenever $\bm{T}$ and $(\boldsymbol{\mu},\boldsymbol{\tau})$ are non-stationary, respectively. 
\item Note that $A^M$ is a continuous map (since the graphical LASSO map is a continuous operator) and the line-search map $B^N$ is also continuous. Since $S = A^M \circ B^N$, it must be continuous as well, from which the closedness of $S$ follows. 
\een

\section{Proof of Theorem 3}
\label{sec:kin}
Fix some spatial coordinate $\bm{x}$ and time-step $t$, and let:
\[\bm{y} = (Y^{(u)}(\bm{x},t;\bm{c}_{new}), Y^{(v)}(\bm{x},t;\bm{c}_{new}), Y^{(w)}(\bm{x},t;\bm{c}_{new}))^T\] be the true simulated flows for $x$-, $y$- and circumferential velocities at the new setting $\bm{c}_{new}$,
\[\hat{\bm{y}} = (\hat{Y}^{(u)}(\bm{x},t;\bm{c}_{new}), \hat{Y}^{(v)}(\bm{x},t;\bm{c}_{new}), \hat{Y}^{(w)}(\bm{x},t;\bm{c}_{new}))^T\]
be its corresponding prediction from (9), and
\[\bar{\bm{y}} = (\bar{Y}^{(u)}(\bm{x};\bm{c}_{new}), \bar{Y}^{(v)}(\bm{x};\bm{c}_{new}), \bar{Y}^{(w)}(\bm{x};\bm{c}_{new}))^T\]
be its time-averaged flow. It is easy to verify that, given the simulation data $\mathcal{D} = \{Y^{(r)}(\bm{x},t;\bm{c}_i)\}$, the conditional distribution of $\bm{y}|\mathcal{D}$ is $\mathcal{N}(\hat{\bm{y}},\Phi(\bm{x},t))$, where:
\begin{equation}
\Phi(\bm{x},t) \equiv
\left[\begin{array}{ccc}
\bm{m}^{(u)} & 0 & 0\\
 0 & \bm{m}^{(v)} & 0\\
 0 & 0 & \bm{m}^{(w)}
 \end{array}\right]
\left[ \mathbb{V}\{{\boldsymbol{\beta}}(t;\bm{c}_{new})| \{\boldsymbol{\beta}(t;\bm{c}_i)\}^n_{i=1}\} \right]_{uvw}\left[\begin{array}{ccc}
 \bm{m}^{(u)} & 0 & 0\\
 0 & \bm{m}^{(v)} & 0\\
 0 & 0 & \bm{m}^{(w)}
 \end{array}\right]^T,
\label{eq:phi}
\end{equation}
with:
\[
\bm{m}^{(r)}=\left[\begin{array}{cccc}
\mathcal{M}_{new} \{\phi^{(r)}_1(\bm{x})\}, & \mathcal{M}_{new} \{\phi^{(r)}_2(\bm{x})\}, & \cdots & \mathcal{M}_{new} \{\phi^{(r)}_{K_r}(\bm{x})\}
\end{array}\right], \quad r=u,v,w. 
\]

Letting $\Phi(t) = \bm{U}\Lambda\bm{U}^T$ be the eigendecomposition of $\Phi(t)$, with $\Lambda = \text{diag}\{\lambda_j\}$, it follows that $\Lambda^{-1/2}\bm{U}^T(\bm{y}-\bar{\bm{y}})|\mathcal{D} \stackrel{d}{=} \mathcal{N}(\boldsymbol{\mu},\bm{I}_K)$, where $\boldsymbol{\mu}=\Lambda^{-1/2}\bm{U}^T(\hat{\bm{y}}-\bar{\bm{y}})$ and $K=K_u+K_v+K_w$. Denoting $\bm{a}=\Lambda^{-1/2}\bm{U}^T(\bm{y}-\bar{\bm{y}})$, the TKE expression in (13) can be rewritten as: 
\begin{align}
\begin{split}
\kappa(\bm{x},t) &= \frac{1}{2}(\bm{y}-\bar{\bm{y}})^T(\bm{y}-\bar{\bm{y}})= \frac{1}{2}(\bm{U}\Lambda^{1/2}\bm{a})^T(\bm{U}\Lambda^{1/2}\bm{a})\\
&= \frac{1}{2}(\bm{a}^T\Lambda^{1/2}\bm{U}^T\bm{U}\Lambda^{1/2}\bm{a})\\
&=\frac{1}{2}\bm{a}^T\Lambda\bm{a}=\frac{1}{2}\sum^K_{j=1}\lambda_ja^2_j.
\end{split}
\end{align}
Since $\bm{a}\sim \mathcal{N}(\boldsymbol{\mu},\bm{I}_K)$, $a^2_j$ has a non-central chi-square distribution with one degree-of-freedom and non-centrality parameter $\mu_j^2$ (we denote this as $\chi^2_1(\mu_j^2)$). $\kappa(\bm{x},t)$ then becomes:
\begin{equation}
\sum^K_{j=1}\frac{\lambda_j}{2}\chi^2_1(\mu^2_j),
\label{eq:dist}
\end{equation}
which is a sum of weighted non-central chi-squared distributions. The computation of the distribution function for such a random variable has been studied extensively, see, e.g., \cite{imhof1961computing}, \cite{davies1973numerical,davies1980algorithm}, \cite{castano2005distribution}, and \cite{liu2009new}, and we appeal to these methods for computing the pointwise confidence interval of $\kappa(\bm{x},t)$ in Section 4. Specifically, we employ the method of \cite{liu2009new} through the \textsf{R} \citep{Rcite} package \texttt{CompQuadForm} \citep{CompQuadForm}.

\end{appendices}

\newpage
\bibliography{references}


\end{document}